\documentclass[twoside,11pt]{article}
\usepackage{versions}

\excludeversion{ICE2008}
\includeversion{ICE2008-TR}

\def\justempty{}
\def\techreport{}
\begin{ICE2008-TR}
\gdef\techreport{yes}
\end{ICE2008-TR}

\ifx\justempty\techreport
\usepackage{entcsmacro}
\usepackage{graphicx}
\usepackage{stmaryrd}
\usepackage{amsmath}
\usepackage{mathrsfs}
\input{macro.tex}
\else
\usepackage{times}
\usepackage{amsmath}
\usepackage{amssymb}
\usepackage{graphicx}
\usepackage{stmaryrd}
\usepackage{mathrsfs}

\makeatletter
\def\@maketitle{%
  \newpage
  \null
  \begin{center}%
  \let \footnote \thanks
    {\LARGE\bf \@title \par}
    \vskip 2em
    {\large
      \lineskip .5em%
      \begin{tabular}[t]{c}%
        \@author
      \end{tabular}}%
    \vskip .3em
  \end{center}%
  \par}
\makeatother

\renewenvironment{abstract}{\begin{list}{}
			{\rightmargin\leftmargin
			\listparindent 1.5em
			\parsep 0pt plus 1pt}
			\small\item}{\end{list}}

\newtheorem{defi}{Definition}[section]
\newtheorem{theo}{Theorem}[section]
\newtheorem{prop}{Proposition}[section]
\newtheorem{lemm}{Lemma}[section]
\newtheorem{coro}{Corollary}[section]
\newtheorem{exam}{Example}[section]

\newenvironment{definition}[2]{\begin{defi} \rm \label{df-#1} \mbox{#2}
\begin{itemize} \vspace{-.7em} \item[]}{\end{itemize}\end{defi}}

\newenvironment{theorem}[2]{\begin{theo} \rm \label{thm-#1} \mbox{#2}
\begin{itemize} \vspace{-.7em} \item[]}{\end{itemize}\end{theo}}

\newenvironment{proposition}[2]{\begin{prop} \rm \label{pr-#1} \mbox{#2}
\begin{itemize} \vspace{-.7em} \item[]}{\end{itemize}\end{prop}}

\newenvironment{propositioni}[2]{\begin{prop} \rm \label{pr-#1} \mbox{#2}
\begin{enumerate} \itemsep -5pt \vspace{-1em}}{\end{enumerate}\end{prop}}

\newenvironment{lemma}[2]{\begin{lemm} \rm \label{lem-#1} \mbox{#2}
\begin{itemize} \vspace{-.7em} \item[]}{\end{itemize}\end{lemm}}

\newenvironment{lemmai}[2]{\begin{lemm} \rm \label{lem-#1} \mbox{#2}
\begin{enumerate} \itemsep -5pt \vspace{-1em}}{\end{enumerate}\end{lemm}}

\newcommand{\refdf}[1]{Definition~\ref{df-#1}}
\newcommand{\refthm}[1]{Theorem~\ref{thm-#1}}
\newcommand{\refpr}[1]{Proposition~\ref{pr-#1}}
\newcommand{\reflem}[1]{Lemma~\ref{lem-#1}}

\newcommand{\reffig}[1]{Figure~\ref{fig-#1}}
\newcommand{\refitem}[1]{(\ref{#1})}

\newenvironment{proof}{\begin{trivlist} \item[\hspace{\labelsep}\bf
    Proof\ \ ]}{\hfill$\boxempty$\end{trivlist}}

\newcommand{\IN}{\mbox{\rm I\hspace{-1.5pt}N}}        
\newcommand{\qed}{\hfill $\Box$ \par \vskip 0.25cm}   
\newfont{\fsc}{eusm10 scaled 1100}		      
\newcommand{\powerset}[1]{\mbox{\fsc P}(#1)}          
\newcommand{\plat}[1]{\raisebox{0pt}[0pt][0pt]{#1}}   

\def\implies{\Rightarrow}

\newcommand{\defitem}[1]{\emph{#1}}                   

\def\alignedcaption[#1&#2]{\mbox{\scriptsize $\mathllap{#1{}}\mathrlap{#2}$}}

\def\structuralN{\mbox{\sf N}}
\def\structuralM{\mbox{\sf M}}
\def\precond#1{{}^\bullet #1}
\def\postcond#1{{#1}^\bullet}
\def\iprecond#1{{}^\circ#1}
\def\ipostcond#1{{#1}^\circ}
\def\Production#1{\stackrel{#1}{\Longrightarrow}}
\def\production#1{\stackrel{#1}{\longrightarrow}}
\def\equivalent{\Leftrightarrow}
\def\into{\rightarrow}
\def\fpair#1{{}\!<\!\!{}#1{}\!\!>\!{}}
\def\failureset{\mathscr{F}}
\def\tb{\tau^\leftarrow}
\def\tbai{\tau^\Leftarrow}

\def\mathrlap{\mathpalette\mathrlapinternal}
\def\mathrlapinternal#1#2{%
  \rlap{$\mathsurround=0pt#1{#2}$}}
\def\mathllap{\mathpalette\mathllapinternal}
\def\mathllapinternal#1#2{%
  \llap{$\mathsurround=0pt#1{#2}$}}
\def\trail#1{\text{~#1}}
\newcommand{\inp}{\mathbin\in}

\makeatletter
\def\goesto{\@transition\rightarrowfill}
\def\Goesto{\@transition\Rightarrowfill}
\def\ngoesto{\@transition\nrightarrowfill}
\def\nGoesto{\@transition\nRightarrowfill}
\def\@transition#1{\@ifnextchar[{\@@transition{#1}}{\@@transition{#1}[]}}
\newbox\@transbox
\newbox\@arrowbox
\def\rightarrowfill{$\m@th\mathord-\mkern-6mu%
  \cleaders\hbox{$\mkern-2mu\mathord-\mkern-2mu$}\hfill
  \mkern-6mu\mathord\rightarrow$}
\def\Rightarrowfill{$\m@th\mathord=\mkern-6mu%
  \cleaders\hbox{$\mkern-2mu\mathord=\mkern-2mu$}\hfill
  \mkern-6mu\mathord\Rightarrow$}
\def\@@transition#1[#2]%
   {\setbox\@transbox\hbox
       {\vrule height 1.5ex depth 1ex width 0ex\hskip0.25em$\scriptstyle#2$\hskip0.25em}
   \ifdim\wd\@transbox<1.5em
      \setbox\@transbox\hbox to 1.5em{\hfil\box\@transbox\hfil}\fi
   \setbox\@arrowbox\hbox to \wd\@transbox{#1}
   \ht\@arrowbox\z@\dp\@arrowbox\z@
   \setbox\@transbox\hbox{$\mathop{\box\@arrowbox}\limits^{\box\@transbox}$}
   \ht\@transbox\z@\dp\@transbox\z@
   \mathrel{\box\@transbox}}
\makeatother

\let\origexists\exists
\let\orignexists\nexists
\let\origforall\forall
\def\exists#1.{\origexists#1.\onespace}
\def\nexists#1.{\orignexists#1.\onespace}
\def\forall#1.{\origforall#1.\onespace}
\def\onespace#1{\let\argument=#1\ifx\onespace#1\else~\fi\argument}

\def\FSI{\text{\it FI}}
\def\FSA{\text{\it FA}}
\def\SI{\text{\it SI}}
\def\SA{\text{\it SA}}
\def\AI{\text{\it AI}}
\def\AA{\text{\it AA}}
\def\FC{{\it FC}}
\def\EFC{{\it EFC}}
\def\BFC{{\it BFC}}
\def\SPL{{\it SPL}}
\def\ESPL{{\it ESPL}}

\def\fsidistance{d}
\def\fsivalidmarking{\alpha}
\def\sidistance{e}

\def\aidistance{f}
\def\aivalidmarking{\gamma}

\let\origmin\min
\def\min{\mathord{\origmin}}
\let\origmax\max
\def\max{\mathord{\origmax}}
\fi

\usepackage{pstricks}
\usepackage{pst-node}
\usepackage{graphicx}

\newcounter{netimage}
\def\p#1:#2;{\cnode #1{0.3}{\thenetimage-#2}}
\def\P#1:#2;{\p #1:#2;\pscircle*#1{0.1}}
\def\q#1:#2:#3;{\p #1:#2;\rput#1{\rput(0.5,0){#3}}}
\def\Q#1:#2:#3;{\P #1:#2;\rput#1{\rput(0.5,0){#3}}}
\def\t#1:#2:#3;{\rput#1{\rnode{\thenetimage-#2}{\psframebox{%
  \vbox to 0.62cm{\vfil\hbox to 0.62cm{\hfil\Large\it \vphantom{tq}#3\hfil}\vspace{-2pt}\vfil}}}}}
\def\a#1->#2;{\ncline{->}{\thenetimage-#1}{\thenetimage-#2}}
\def\A#1->#2;{\ncarc{->}{\thenetimage-#1}{\thenetimage-#2}}

\makeatletter
\long\def\petrinet(#1)#2\end{\psscalebox{0.7}{\pspicture(#1)\stepcounter{netimage}#2\endpspicture}\end}

\makeatother

\def\lastname{van Glabbeek, Goltz and Schicke}
\begin{ICE2008-TR}
\def\titleheader{Symmetric and Asymmetric Asynchronous Interaction}
\title{Symmetric and Asymmetric Asynchronous Interaction}
\author{Rob van Glabbeek\\
 \footnotesize NICTA, Sydney, Australia\\[-3pt]
 \footnotesize University of New South Wales, Sydney, Australia\\
 \footnotesize \tt rvg@cs.stanford.edu
\and\hspace{-1em}
    Ursula Goltz
\hspace{4em}
  Jens-Wolfhard Schicke\thanks{This paper was partially
  written during a four month stay of J.-W. Schicke at NICTA,
  during which he was supported by DAAD (Deutscher Akademischer Austauschdienst) and NICTA.
}\\
 \footnotesize Institute for Programming and Reactive Systems\\[-3pt]
 \footnotesize TU Braunschweig, Braunschweig, Germany\\
 \footnotesize \tt 
\hspace{-4.5em}goltz@ips.cs.tu-bs.de
\hspace{3.5em} drahflow@gmx.de}

\end{ICE2008-TR}
\begin{ICE2008}
\gdef\nictaref{}
\gdef\nicta{}
\end{ICE2008}

\begin{document}

\begin{ICE2008-TR}
\maketitle
\thispagestyle{empty}
\begin{abstract} 
  We investigate classes of systems based on different interaction patterns
  with the aim of achieving distributability. As our system model we use Petri
  nets. In Petri nets, an inherent concept of simultaneity is built in, since
  when a transition has more than one preplace, it can be crucial that tokens
  are removed instantaneously.  When modelling a system which is intended to be
  implemented in a distributed way by a Petri net, this built-in concept of
  synchronous interaction may be problematic. To investigate the problem we
  assume that removing tokens from places can no longer be considered as
  instantaneous. We model this by inserting silent (unobservable) transitions
  between transitions and their preplaces. We investigate three different
  patterns for modelling this type of asynchronous interaction. \textit{Full
  asynchrony} assumes that every removal of a token from a place is
  time consuming. For \textit{symmetric asynchrony}, tokens are only removed
  slowly in case of backward branched transitions, hence where the concept of
  simultaneous removal actually occurs. Finally we consider a more intricate
  pattern by allowing to remove tokens from preplaces of backward branched
  transitions asynchronously in sequence (\textit{asymmetric asynchrony}). 

  We investigate the effect of these different transformations of instantaneous
  interaction into asynchronous interaction patterns by comparing the
  behaviours of nets before and after insertion of the silent transitions. We
  exhibit for which classes of Petri nets we obtain equivalent behaviour with
  respect to failures equivalence. 
  
  It turns out that the resulting hierarchy of Petri net classes can be
  described by semi-structural properties. In case of full
  asynchrony and symmetric asynchrony, we obtain precise characterisations;
  for asymmetric asynchrony we obtain lower and upper bounds. 
  
  We briefly comment on possible applications of our results to Message
  Sequence Charts.
\end{abstract}
\end{ICE2008-TR}

\begin{ICE2008}
\begin{frontmatter}
  \title{Symmetric and Asymmetric\\ Asynchronous Interaction}
  \author{Rob van Glabbeek\thanksref{robemail}}
  \address{NICTA, Sydney, Australia}
  \address{University of New South Wales, Sydney, Australia}
  \author{Ursula Goltz\thanksref{ullaemail}}
  \author{Jens-Wolfhard Schicke\thanksref{jensemail}}
  \address{Institute for Programming and Reactive Systems\\TU Braunschweig\\
  Braunschweig, Germany\nictaref}
  \thanks[robemail]{Email: \href{mailto:rvg@cs.stanford.edu} {\texttt{\normalshape rvg@cs.stanford.edu}}}
  \thanks[ullaemail]{Email: \href{mailto:goltz@ips.cs.tu-bs.de} {\texttt{\normalshape goltz@ips.cs.tu-bs.de}}}
  \thanks[jensemail]{Email: \href{mailto:drahflow@gmx.de} {\texttt{\normalshape drahflow@gmx.de}}}
  \nicta
\begin{abstract} 
  We investigate classes of systems based on different interaction patterns
  with the aim of achieving distributability. As our system model we use Petri
  nets. In Petri nets, an inherent concept of simultaneity is built in, since
  when a transition has more than one preplace, it can be crucial that tokens
  are removed instantaneously.  When modelling a system which is intended to be
  implemented in a distributed way by a Petri net, this built-in concept of
  synchronous interaction may be problematic. To investigate the problem we
  assume that removing tokens from places can no longer be considered as
  instantaneous. We model this by inserting silent (unobservable) transitions
  between transitions and their preplaces. We investigate three different
  patterns for modelling this type of asynchronous interaction. \textit{Full
  asynchrony} assumes that every removal of a token from a place is
  time consuming. For \textit{symmetric asynchrony}, tokens are only removed
  slowly in case of backward branched transitions, hence where the concept of
  simultaneous removal actually occurs. Finally we consider a more intricate
  pattern by allowing to remove tokens from preplaces of backward branched
  transitions asynchronously in sequence (\textit{asymmetric asynchrony}). 

  We investigate the effect of these different transformations of instantaneous
  interaction into asynchronous interaction patterns by comparing the
  behaviours of nets before and after insertion of the silent transitions. We
  exhibit for which classes of Petri nets we obtain equivalent behaviour with
  respect to failures equivalence. 
  
  It turns out that the resulting hierarchy of Petri net classes can be
  described by semi-structural properties. In case of full
  asynchrony and symmetric asynchrony, we obtain precise characterisations;
  for asymmetric asynchrony we obtain lower and upper bounds. 
  
  We briefly comment on possible applications of our results to Message
  Sequence Charts.
\end{abstract}
\end{frontmatter}
\setcounter{footnote}{0}
\vspace{-1em}
\end{ICE2008}

\section{Introduction}\label{sec-intro}
  In this paper, we investigate classes of systems based on different
  asynchronous interaction patterns with the aim of achieving
  distributability, i.e.\ the possibility to execute a system on
  spatially distributed locations, which do not share a common clock.
  As our system model we use Petri nets. The main
  reason for this choice is the detailed way in which a Petri net
  represents a concurrent system, including the interaction between
  the components it may consist of. In an interleaving based model of
  concurrency such as labelled transition systems modulo bisimulation
  semantics, a system representation as such cannot be said to display
  synchronous or asynchronous interaction; at best these are
  properties of composition operators, or communication primitives,
  defined in terms of such a model. A Petri net on the other hand
  displays enough detail of a concurrent system to make the presence
  of synchronous communication discernible. This makes it possible to
  study asynchronous communication without digressing to the realm of
  composition operators.

  In a Petri net, a transition interacts with
  its preplaces by consuming tokens. An inherent concept of simultaneity is
  built in, since when a transition has more than one preplace, it can be
  crucial that tokens are removed instantaneously, depending on the surrounding
  structure or---more elaborately---the behaviour of the net. 
  
  When modelling a distributed system by a Petri net, this built-in concept of
  synchronous interaction may become problematic. Assume a transition $t$ on a
  location $l$ models an activity involving another location $l'$, for example by
  receiving a message. This can be modelled by a preplace  $s$ of $t$ such that $s$
  and $t$ are situated in different locations. We assume that taking a token can
  in this situation not be considered as instantaneous; rather the interaction
  between $s$ and $t$ takes time. We model this effect by inserting silent
  (unobservable) transitions between transitions and their
  preplaces. We call the effect of such a transformation of a net $N$
  an \textit{asynchronous implementation} of $N$.
  
\begin{figure}
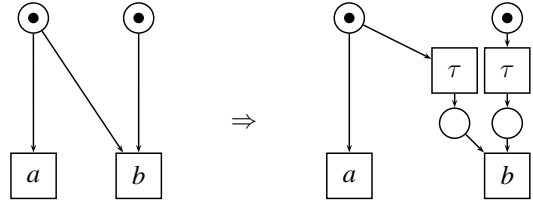

  \begin{center}
    \begin{petrinet}(13,4.5)
      \P (2,4):p1;
      \P (4,4):p2;
      \t (2,1):t1:a;
      \t (4,1):t2:b;
      \a p1->t1;
      \a p1->t2;
      \a p2->t2;
      \a p2->n1;

      \rput(6,2){\Large $\Rightarrow$}

      \P (8,4):p1b;
      \P (11,4):p2b;
      \t (11,3):p2btt:$\tau$;
      \p (11,2):p2btp;
      \t (10,3):p1bt2t:$\tau$;
      \p (10,2):p1bt2p;
      \t (8,1):t1b:a;
      \t (11,1):t2b:b;
      \a p1b->t1b;
      \a p2b->p2btt; \a p2btt->p2btp; \a p2btp->t2b;
      \a p1b->p1bt2t; \a p1bt2t->p1bt2p; \a p1bt2p->t2b;
      \a p2b->n1b;
    \end{petrinet}
  \end{center}
\vspace{-1em}
  \caption{Transformation to the symmetrically asynchronous implementation}
  \label{fig-sym-transform}
\end{figure}

  An example of such an implementation is shown in \reffig{sym-transform}. Note
  that $a$ can be disabled in the implementation before any visible behaviour
  has taken place. This difference will cause non-equivalence between the
  original and the implementation under branching time equivalences.

  Our asynchronous implementation allows a token to start its journey
  from a place to a transition even when not all preplaces of the
  transition contain a token. This design decision is motivated by the
  observation that it is fundamentally impossible to check in an
  asynchronous way whether all preplaces of a transition are marked---it
  could be that a token moves back and forth between two such places.

  We investigate different interaction patterns for the asynchronous
  implementation of nets. The simplest pattern (\textit{full asynchrony})
  assumes that every removal of a token from a place is time consuming. For the next pattern
  (\textit{symmetric asynchrony}), tokens are only removed slowly when they are consumed
  by a backward branched transition, hence where the concept of simultaneous
  removal actually occurs. Finally we consider a more intricate pattern by
  allowing to remove tokens from preplaces of backward branched transitions
  asynchronously in sequence (\textit{asymmetric asynchrony}).

Given a choice of interaction pattern, we call a net $N$
\defitem{asynchronous} when there is no essential behavioural
difference between $N$ and its asynchronous implementation $I(N)$.
In order to formally define this concept, we wish to compare the
behaviours of $N$ and $I(N)$ using a semantic equivalence that
fully preserves branching time, causality and their interplay, whilst
of course abstracting from silent transitions.  By choosing the most
discriminating equivalence possible, we obtain the smallest possible class of
asynchronous nets, thus excluding nets that might be classified as
asynchronous merely because a less discriminating equivalence would
fail to see the differences between such a net and its asynchronous
implementation.
To simplify the exposition, here we merely compare the behaviours of
$N$ and $I(N)$ up to \defitem{failures equivalence} \cite{BHR84}.
This interleaving equivalence abstracts from causality and respects
branching time only to some degree. However, we conjecture that our
results are in fact largely independent of this choice and that more
discriminating equivalences, such as the history preserving
ST-bisimulation of \cite{Vo93}, would yield the same classes of
asynchronous nets. Using a linear time equivalence would give rise to
larger classes; this possibility is investigated in
\cite{schicke08synchrony}.

Thus we investigate the effect of our three transformations of
instantaneous interaction into asynchronous interaction patterns by
comparing the behaviours of nets before and after insertion of the
silent transitions up to failures equivalence.  We show that in the
case of full asynchrony, we obtain equivalent behaviour
exactly for conflict-free Petri nets. Further we establish that
symmetric asynchrony is a valid concept for $\structuralN$-free Petri
nets and asymmetric asynchrony for $\structuralM$-free Petri nets,
where $\structuralN$ and $\structuralM$ stand for certain structural
properties; the reachability of such structures is crucial. For
symmetric asynchrony we obtain a precise characterisation of the class
of nets which is asynchronously implementable. For asymmetric
asynchrony we obtain lower and upper bounds.

In the concluding section, we discuss the use of our results for
Message Sequence Charts, as an example how they may be useful for
other models than Petri nets.  When interpreting basic Message
Sequence Chart as Petri nets, the resulting Petri nets lie within the
class of conflict-free and hence $\structuralN$-free Petri nets.  The
more expressive classes give insights in the effect of choices in
non-basic MSCs.

\begin{ICE2008}
This is an extended abstract; for sake of brevity most proofs are omitted.
They are contained in the full version of this paper \cite{glabbeek08symmasymm},
as well as in \cite{schicke08synchrony}.
\end{ICE2008}

The paper is structured as follows. In Section \ref{sec-notions} we establish
the necessary basic notions. In Section \ref{sec-fsa} we introduce the fully
asynchronous transformation and give a semi-structural
characterisation of the resulting net class. In Section \ref{sec-sa} we repeat
those steps for the symmetrically asynchronous transformation. Furthermore we
describe how the resulting net class relates to the classes of
free-choice and extended free choice nets. In Section \ref{sec-aa} we introduce
the asymmetrically asynchronous transformation. We give semi-structural upper
and lower bounds for the resulting net class and relate it to simple and
extended simple nets. In the conclusion in Section~\ref{sec-conclusion} we
compare our findings to similar results in the literature.

\begin{ICE2008-TR}
An extended abstract of this paper will be presented at the first {\sl
Interaction and Concurrency Experience} (ICE'08) on {\sl Synchronous
and Asynchronous Interactions in Concurrent Distributed Systems}, and will
appear in {\sl Electronic Notes in Theoretical Computer Science}, Elsevier.
\end{ICE2008-TR}

\section{Basic Notions}\label{sec-notions}

We consider here 1-safe net systems, i.e.\ places never carry more than
one token, but a transition can fire even if pre- and postset intersect.
To represent unobservable behaviour, which we use to model asynchrony,
the set of transitions is partitioned into observable and silent (unobservable) ones.

\begin{definition}{nst}{}
  A \defitem{net with silent transitions} is a tuple
  $N = (S, O, U, F, M_0)$ where
  \begin{itemize}
    \item $S$ is a set (of \defitem{places}),
    \item $O$ is a set (of \defitem{observable transitions}),
    \item $U$ is a set (of \defitem{silent transitions}),
    \item $F \subseteq S \times T \cup T \times S$
      (the \defitem{flow relation}) with $T := O \cup U$
      (\defitem{transitions}) and
    \item $M_0 \subseteq S$ (the \defitem{initial marking}).
  \end{itemize}
\end{definition}

Petri nets are depicted by drawing the places as circles, the
transitions as boxes, and the flow relation as arrows (\defitem{arcs})
between them.  When a Petri net represents a concurrent system, a
global state of such a system is given as a \defitem{marking}, a set
of places, the initial state being $M_0$.  A marking is depicted by
placing a dot (\defitem{token}) in each of its places.  The dynamic
behaviour of the represented system is defined by describing the
possible moves between markings. A marking $M$ may evolve into a
marking $M'$ when a nonempty set of transitions $G$ \defitem{fires}. In that
case, for each arc $(s,t) \in F$ leading to a transition $t$ in $G$, a
token moves along that arc from $s$ to $t$.  Naturally, this can
happen only if all these tokens are available in $M$ in the first
place. These tokens are consumed by the firing, but also new tokens
are created, namely one for every outgoing arc of a transition in
$G$. These end up in the places at the end of those arcs.  A problem
occurs when as a result of firing $G$ multiple tokens end up in the
same place. In that case $M'$ would not be a marking as
defined above. In this paper we restrict attention to nets in which
this never happens. Such nets are called \defitem{1-safe}.
Unfortunately, in order to formally define this class of nets, we
first need to correctly define the firing rule without assuming
1-safety. Below we do this by forbidding the firing of sets of
transitions when this might put multiple tokens in the same place.

\begin{definition}{steps}{
  Let $N = (S, O, U, F, M_0)$ be a net. Let $M_1, M_2 \subseteq S$.
  }
We denote the preset and postset of a net element $x$ by
$\precond{x} := \{y \mid (y, x) \in F\}$
and
$\postcond{x} := \{y \mid (x, y) \in F\}$
respectively.
A nonempty set of transitions $G \subseteq (O \cup U),
  G \not= \varnothing$, is called a \defitem{step from $M_1$ to $M_2$},
  notation $M_1 ~[G\rangle_N~ M_2$, iff
  \begin{itemize}
    \item all transitions contained in $G$ are \defitem{enabled}, that is
      \begin{equation*}
        \forall t\in G. \precond{t} \subseteq M_1 \wedge 
          (M_1 \setminus \precond{t}) \cap \postcond{t} =
          \varnothing \trail{,}
      \end{equation*}
    \item all transitions of $G$ are \defitem{independent}, that is \defitem{not conflicting}:
      \begin{align*}
        &\forall t,u \in G, t\not= u. \precond{t} \cap \precond{u} = \varnothing
        \wedge \postcond{t} \cap \postcond{u} = \varnothing \trail{,}
      \end{align*}
    \item in $M_2$ all tokens have been removed from the \defitem{preplaces}
      of $G$ and new tokens have been inserted at the
      \defitem{postplaces} of $G$:
      \begin{equation*}
        M_2 = \left(M_1 \setminus \bigcup_{t\in G} \precond{t}\right) \cup
        \bigcup_{t\in G}\postcond{t} \trail{.}
      \end{equation*}
  \end{itemize}
\end{definition}

To simplify statements about possible behaviours of nets, we use some abbreviations.

\begin{definition}{steprel}{Let $N = (S, O, U, F, M_0)$ be a net with silent transitions.}
  \begin{itemize}
    \item $\mathord{\production{}_N} \subseteq \powerset{S} \times \powerset{O} \times
      \powerset{S}$ is defined by $M_1 \production{G}_N M_2 \equivalent
      G \mathbin{\subseteq} O \wedge M_1[G\rangle_N\, M_2$
    \item $\mathord{\production{\tau}_N} \subseteq \powerset{S} \times \powerset{S}$
      is defined by $M_1 \production{\tau}_N M_2 \equivalent
      \exists t \inp U. M_1 \,[\{t\}\rangle_N~ M_2$
    \item $\mathord{\Production{}_N} \subseteq \powerset{S} \times O^* \times
      \powerset{S}$ is defined by
      $M_1 \Goesto[\,t_1 t_2 \cdots t_n~]_N M_2 \equivalent\\[3pt]
      \hphantom{M_1 \Production{\sigma}_N M_2 \equivalent }
      M_1
      \production{\tau}^*_N \production{\{t_1\}}_N
      \production{\tau}^*_N \production{\{t_2\}}_N
      \production{\tau}^*_N \cdots
      \production{\tau}^*_N \production{\{t_n\}}_N
      \production{\tau}^*_N
      M_2$\\[3pt]
      where $\production{\tau}^*_N$ denotes the reflexive and
      transitive closure of $\production{\tau}_N$.
  \end{itemize}

  We write $M_1 \production{G}_N$ for $\exists M_2. M_1
  \production{G}_N M_2$, $M_1 \arrownot\production{G}_N$ for $\nexists M_2. M_1
  \production{G}_N M_2$ and similar for the other two relations. 

  A marking $M_1$ is said to be \defitem{reachable} iff there is a
  $\sigma \in O^*$ such that $M_0 \Production{\sigma} M_1$. The set of all
  reachable markings is denoted by $[M_0\rangle_N$.
\end{definition}

We omit the subscript $N$ if clear from context.

As said before, here we only want to consider 1-safe nets. Formally,
we restrict ourselves to \defitem{contact-free nets} where in every
reachable marking $M_1 \in [M_0\rangle$ for all $t \in O \cup U$ with
  $\precond{t} \subseteq M_1$\begin{ICE2008}\vspace{-1ex}\end{ICE2008}
\begin{equation*}
  (M_1 \setminus \precond{t}) \cap \postcond{t} = \varnothing \trail{.}
\begin{ICE2008}\vspace{1ex}\end{ICE2008}\end{equation*}
For such nets, in \refdf{steps} we can just as well consider a
transition $t$ to be enabled in $M$ iff $\precond{t}\subseteq M$, and
two transitions to be independent when $\precond{t} \cap \precond{u} =
\varnothing$.
In this paper we furthermore restrict attention to nets for which
$\precond{t}\neq\emptyset$, and $\precond{t}$ and $\postcond{t}$ are
finite for all $t\inp O\cup U$.  We also require the initial marking $M_0$ to
be finite.\begin{ICE2008}\linebreak\end{ICE2008}
A consequence of these restrictions is that all reachable
markings are finite, and it can never happen that infinitely many
independent transitions are enabled. Henceforth, we employ the name
\defitem{$\tau$-nets} for nets with silent transitions obeying the
above restrictions, and \defitem{plain nets} for $\tau$-nets without
silent transitions, i.e.\ with $U=\emptyset$.

\begin{ICE2008-TR}
Plain nets have the nice property of being deterministic, i.e.\ the marking
obtained after firing a sequence of transitions is uniquely determined
by the sequence of transitions fired.

\begin{lemma}{determ}{
  Let $N = (S, O, \varnothing, F, M_0)$ be a plain net, $\sigma \inp O^*$ and $M \subseteq S$.
  }
  If $M \Production{\sigma} M_1 \wedge M \Production{\sigma} M_2$ then
  $M_1 = M_2$.
\end{lemma}
\begin{proof}
  Let $t \inp O$, $\sigma \inp O^*$ and $M' \subseteq S$.
\\
  Then
  $M \Production{\{t\}} M' \equivalent M \production{\{t\}} M'$ and
  $M \production{\{t\}} M'$ implies $M' = (M \setminus \precond{t}) \cup
  \postcond{t}$.
\\
  Hence $M \Production{\{t\}} M_1 \wedge M \Production{\{t\}} M_2$ implies
  $M_1 = M_2$.
\\
  The result follows for a trace $\sigma$ by induction on the length of $\sigma$.
\end{proof}
\end{ICE2008-TR}

Our nets with silent transitions can be regarded as special
\defitem{labelled nets}, defined as in
Definition~\ref{df-nst}, but without the split of $T$ into $O$ and
$U$, and instead equipped with a \defitem{labelling function} $\ell: T
\rightarrow {\it Act} \cup \{\tau\}$, where {\it Act} is a set of
\defitem{visible actions} and $\tau\not\in{\it Act}$ an invisible
one. Nets with silent transitions correspond to labelled nets in which
no two different transitions are labelled by the same visible actions,
which can be formalised by taking $\ell(t)=t$ for $t\in O$ and
$\ell(t)=\tau$ for $t\in U$.

To describe which nets are ``asynchronous'', we will compare their
behaviour to that of their asynchronous implementations using a
suitable equivalence relation. As explained in the introduction, we
consider here branching time semantics.  Technically, we use failures
equivalence, as defined below.

\begin{definition}{failurepair}{
  Let $N = (S, O, U, F, M_0)$ be a $\tau$-net, $\sigma \in O^*$ and
  $X \subseteq O$.
  }
  $\fpair{\sigma,X}$ is a \defitem{failure pair} of $N$ iff\vspace{-1em}
  $$
    \exists M_1. M_0 \Production{\sigma} M_1 \wedge M_1
    \arrownot\production{\tau} \wedge \forall t \in X. M_1
    \arrownot\production{\{t\}} \trail{.}
  \vspace{-7pt}$$
  We define $\failureset(N) := \{\fpair{\sigma,X} \mid
  \fpair{\sigma,X} \text{ is a failure pair of } N\}$.

  Two $\tau$-nets $N$ and $N'$ are \defitem{failures equivalent}, $N
  \approx_{\failureset} N'$, iff $\failureset(N) = \failureset(N')$.
\end{definition}

A $\tau$-net $N = (S, O, U, F, M_0)$ is called
\defitem{divergence free} iff there are no infinite chains of markings
$M_1 \production{\tau} M_2 \production{\tau} \cdots$ with $M_1 \in
[M_0\rangle$.

\section{Full Asynchrony}\label{sec-fsa}

As explained in the introduction, we will examine in this paper
different possible assumptions of how asynchronous interaction between
transitions and their preplaces takes place. In this section, we
start with the simple and intuitive assumption that the removal of any
token by a transition takes time. This is implemented by inserting
silent transitions between visible ones and their preplaces. 

\begin{definition}{fsi}{
  Let $N = (S, O, \varnothing, F, M_0)$ be a plain net.
  }
  The \defitem{fully asynchronous implementation} of $N$ is
  defined as the net\\
  $\FSI(N) := (S \cup S^\tau, O, U', F', M_0)$ with\vspace{-1ex}
  \begin{align*}
    S^\tau :=~& \{s_t \mid t \in O, s \in \precond{t}\} \trail{,}\\
    U' :=~& \{t_s \mid t \in O, s \in \precond{t}\} \trail{and}\\
    F' :=~& (F \cap (O \times S))
      \cup \{(s, t_s), (t_s, s_t), (s_t, t) \mid t \in O, s \in \precond{t}\}
      \trail{.}
  \end{align*}
\end{definition}

\begin{ICE2008-TR}

For better readability we will use the abbreviations $\iprecond{x} := \{y \mid (y,
x) \in F'\}$ and $\ipostcond{x} := \{y \mid (x, y) \in F'\}$ instead of
$\precond{x}$ or $\postcond{x}$ when making assertions about the flow relation
of an implementation.

The following lemma shows how the fully asynchronous implementation of
a plain net $N$ simulates the behaviour of $N$.

\begin{lemmai}{fsisim}{
  Let $N = (S, O, \varnothing, F, M_0)$ be a plain net, $G\subseteq
  O$, $\sigma\in O^*$ and $M_1, M_2 \subseteq S$.
  }
  \item If $M_1 \production{G}_{N} M_2$ then
  $M_1 \production{\tau}^*_{\FSI(N)}\production{G}_{\FSI(N)} M_2$.
  \item If $M_1 \Production{\sigma}_{N} M_2$ then $M_1 \Production{\sigma}_{\FSI(N)} M_2$.
\vspace{-1ex}
\end{lemmai}

\begin{proof}
  Assume $M_1 \production{G}_N M_2$.  Then, by construction of $\FSI(N)$,
  $$M_1 ~[\{t_s \mid t \inp G,~ s \inp \precond{t}\}\rangle_{\FSI(N)}~
  [\{t \mid t \inp G\}\rangle_{\FSI(N)}~ M_2.$$
  The first part of that execution can be split into a sequence of singletons.
\\
  The second statement follows by a straightforward induction on the
  length of $\sigma$.
\end{proof}

This lemma uses the fact that any marking of $N$ is also a marking on
$\FSI(N)$. The reverse does not hold, so in order to describe the degree
to which the behaviour of $\FSI(N)$ is simulated by $N$ we need to
explicitly relate markings of $\FSI(N)$ to those of $N$.
This is in fact not so hard, as any reachable marking of $\FSI(N)$ can be
obtained from a reachable marking of $N$ by moving some tokens into
the newly introduced buffering places $s_t$. To establish this
formally,
we define a function which transforms implementation markings into the
related original markings, by shifting these tokens back.

\begin{definition}{tauback}{
  Let $N = (S, O, \varnothing, F, M_0)$ be a plain net
  and let $\FSI(N) = (S \cup S^\tau, O, U', F', M_0)$.
  }
  $\tb: S \cup S^\tau \into S$ is the function defined by
  \begin{equation*}
    \tb(p) := \begin{cases}
      s & \text{ iff } p = s_t \text{ with } s_t \in S^\tau, s \in S, t \in O\\
      p & \text{ otherwise } (p \in S)
    \end{cases}
  \end{equation*}
\end{definition}

Where necessary we extend functions to sets elementwise.
So for any $M \subseteq S \cup S^\tau$ we have $\tb(M)=\{\tb(s)\mid
s\in M\} = (M\cap S)\cup \{s \mid s_t\in M\}$. In particular,
$\tb(M)=M$ when $M\subseteq S$.

We now introduce a predicate $\alpha$ on the markings of $\FSI(N)$
that holds for a marking iff it can be obtained from a reachable
marking of $N$ (which is also a marking of $\FSI(N)$) by firing some
unobservable transitions. Each of these unobservable transitions moves
a token from a place $s$ into a buffering place $s_t$.\linebreak[2]
Later, we will show that $\alpha$ exactly characterises the reachable
markings of $\FSI(N)$.
Furthermore, as every token can be moved only once, we can also give
an upper bound on how many such movements can still take place.

\begin{definition}{fsivalidmarking}{
  Let $N = (S, O, \varnothing, F, M_0)$ be a plain net and $\FSI(N) = (S \cup S^\tau,
  O, U', F', M_0)$.
  }
  The predicate $\fsivalidmarking \subseteq \powerset{S \cup S^\tau}$ is given by
  $$\fsivalidmarking(M) ~~:\equivalent~~
  \tb(M) \inp [M_0\rangle_N \wedge \forall p, q \inp M. \tb(p) = \tb(q) \implies p = q.$$
  The function $\fsidistance: \powerset{S \cup S^\tau} \into \IN \cup \{\infty\}$ is given by
  $\fsidistance(M) := |M \cap \{s \mid s \in S,~ \postcond{s} \ne \varnothing\}|$, where
  we choose not to distinguish between different degrees of infinity.
\end{definition}

Note that $\fsivalidmarking(M)$ implies $|M|=|\tb(M)|$, and reachable
markings of $N$ are always finite (thanks to our definition of a plain
net).  Hence $\fsivalidmarking(M)$ implies $d(M)\in\IN$.  The
following lemma confirms that our informal description of $\fsivalidmarking$
matches its formal definition.

\begin{lemma}{fsivalidmarking}{
  Let $N$ and $\FSI(N)$ be as above and $M \subseteq S\cup S^\tau$,
  with $M$ finite.
  }
  Then $\forall p, q \inp M. \tb(p) = \tb(q) \implies p = q$ iff
  $\tb(M)\production{\tau}_{\FSI(N)}^* M$.
\end{lemma}

\begin{proof} 
  Given that $\tb(M)\subseteq S$, ``if'' follows directly from the construction
  of $\FSI(N)$.\\
  For ``only if'', assume $\forall p, q \inp M. \tb(p) = \tb(q) \implies p = q$.
  Then $\tb(M) ~[\{t_s \mid s_t \inp M\}\rangle_{\FSI(N)}~ M$.
\end{proof}

Now we can describe how any net simulates the behaviour of
its fully asynchronous implementation.

\begin{lemmai}{fsiworking}{
  Let $N$ and $\FSI(N)$ be as above, $G\subseteq O$, $\sigma\in O^*$ and $M,M' \subseteq S \cup S^\tau$.
  }
    \item\label{fsiinvarstart} $\fsivalidmarking(M_0)$.
    \item\label{fsiimplstep} If $\fsivalidmarking(M) \wedge M
      \production{G}_{\FSI(N)} M'$ then
      $\tb(M) \production{G}_N \tb(M') \wedge \fsivalidmarking(M')$.
    \item\label{fsiimpltaustep} If $\fsivalidmarking(M) \wedge M \production{\tau}_{\FSI(N)} M'$
    then $\fsidistance(M) > \fsidistance(M') \wedge \tb(M) = \tb(M') \wedge \fsivalidmarking(M')$.
    \item\label{fsiimpllongstep} If $M_0 \Production{\sigma}_{\FSI(N)}
      M'$ then $M_0 \Production{\sigma}_N \tb(M')\wedge \fsivalidmarking(M')$.
\end{lemmai}
\begin{proof}
  \refitem{fsiinvarstart}: $M_0 \in [M_0\rangle_N$ and $\forall s \in M_0 \subseteq S. \tb(s) = s$.

  \refitem{fsiimplstep}: Suppose $\fsivalidmarking(M)$ and $M
    \production{G}_{\FSI(N)} M'$ with $G\subseteq O$.
    So $\tb(M)$ is a reachable marking of $N$.

    Let $t \in G$. Since $t$ is enabled in $M$, we have
    $\iprecond{t} \subseteq M$ and hence $\tb(\iprecond{t}) \subseteq
    \tb(M)$.  By construction, $\iprecond{t} = \{s_t \mid s \in
    \precond{t}\}$ so $\tb(\iprecond{t}) = \precond{t}$.  Given that
    $N$ is contact-free, it follows that $t$ is enabled in $\tb(M)$.

    Now let $t,u \in G$ with $t\neq u$.  If $s \in \precond{t}
    \cap \precond{u}$ then $s_t \in \iprecond{t}$ and $s_u \in
    \iprecond{u}$, so $s_t,s_u \in M$.  However, $\tb(s_t)=\tb(s_u)$,
    contradicting $\fsivalidmarking(M)$. Hence $\precond{t} \cap
    \precond{u} = \emptyset$. Given that $\precond{t} \cup
    \precond{u} \subseteq \tb(M)$ and $N$ is contact-free, it
    follows that also $\postcond{t}\cap\postcond{u}=\emptyset$ and
    hence $t$ and $u$ are independent.

    We will now show that $\displaystyle
    \left(\tb(M)\setminus\bigcup_{t\in G} \precond{t}\right) \cup
        \bigcup_{t\in G}\postcond{t} = \tb(M')$.
    \begin{align*}
      M'
      ={}& (M \setminus \{s \mid s \in \iprecond{t},~ t \in G\}) \cup \{s \mid s \in
      \ipostcond{t},~ t \in G\}\\
      ={}& (M \setminus \{s_t \mid s \in \precond{t},~ t \in G\})
      \cup \{ s \mid s \in \postcond{t},~ t \in G\}
      \trail{.}
    \end{align*}
    Therefore
       $\tb(M') = \tb(M \setminus \{s_t \mid s \in \precond{t},~ t \in G\})
      \cup \tb(\{ s \mid s \in \postcond{t},~ t \in G\})$.

    Take any $t \in G$ and any $s \in \precond{t}$. Then $s_t \in M$ and
    $\fsivalidmarking(M)$ implies $s \mathbin{\notin} M \wedge \nexists
    u \inp O. u \ne t \wedge s_u \inp M$.
    Hence $\tb(M \setminus \{s_t \mid s \in \precond{t},~ t \in G\}) =
    \tb(M) \setminus \{s \mid s \in \precond{t},~ t \in G\}$.
    Thus we find
    \begin{align*}
      \tb(M')
      ={}& \tb(M) \setminus \{s \mid s \in \precond{t},~ t \in G\}
      \cup \{ s \mid s \in \postcond{t},~ t \in G\}
    \end{align*}
    and conclude that $\tb(M) \production{G}_N \tb(M')$.

    Next we establish $\fsivalidmarking(M')$. To this end, we may
    assume that $G$ is a singleton set, for $G$ must be finite---this
    follows from our definition of a plain net---and when
    $M[\{t_0,t_1,\ldots,t_n\}\rangle M'$ for some $n\geq 0$ then there
    are $M_1,M_2,\ldots,M_n$ with $M[\{t_0\}\rangle M_1[\{t_1\}\rangle M_2
    \cdots M_n[\{t_n\}\rangle M'$, allowing us to obtain the general
    case by induction.  So let $G=\{t\}$ with $t\inp T$.

    Above we have shown that $\tb(M')\in[M_0\rangle_N$.
    We still need to prove that $\forall p, q \in M'. p \ne q \implies \tb(p)
    \ne \tb(q)$.
    Assume the contrary, i.e.\ there are $p, q \in M'$ with $p \ne q \wedge
    \tb(p) = \tb(q)$.
    Since $\fsivalidmarking(M)$, at least one of $p$ and $q$---say $p$---must not be
    present in $M$.
     Then $p \in \ipostcond{t}=\postcond{t}\subseteq S$.
    As $\tb(q)=\tb(p)=p$ and $q\neq p$, it must be that $q\in S^\tau$.
    Hence $q\mathbin{\notin}\ipostcond{t}$, so $q\inp M$, and $p=\tb(q)\in\tb(M)$.
    As shown above, $t$ is enabled in $\tb(M)$.
    By the contact-freeness of $N$, $(\tb(M)\setminus
    \precond{t})\cap \postcond{t}=\emptyset$, so $p \inp \precond{t}$.
    Hence $p_t \in \iprecond{t} \subseteq M$. As by construction
    $\iprecond{t}\cap\ipostcond{t}=\emptyset$, we have $p_t\not\in
    M'$, so $q\neq p_t$. Yet $\tb(q)=\tb(p_t)$, contradicting $\fsivalidmarking(M)$.

  \refitem{fsiimpltaustep}: Let $t_s \in U'$ such that $M [\{t_s\}\rangle_{\FSI(N)} M'$.
    Then, by construction of $\FSI(N)$, $\iprecond{t_s} = \{s\} \wedge \ipostcond{t_s} = \{s_t\}$.
    Hence $M' = M \setminus \{s\} \cup \{s_t\}$ and
    $\fsidistance(M') = \fsidistance(M) - 1 \wedge \tb(M') = \tb(M)$.
    Moreover, $\fsivalidmarking(M') \equivalent \fsivalidmarking(M)$.

  \refitem{fsiimpllongstep}: Using (\ref{fsiinvarstart}--\ref{fsiimpltaustep}),
   this follows by a straightforward induction on the number of
   transitions in the derivation $M_0 \Production{\sigma}_{\FSI(N)} M'$.
\end{proof}
\end{ICE2008-TR}

\begin{ICE2008}
It is not hard to see that implementations of contact-free nets are
contact-free and implementations are always divergence free; in fact
an implementation of a plain net is always a divergence free $\tau$-net.
\end{ICE2008}
\begin{ICE2008-TR}
It follows that $\fsivalidmarking$ exactly characterises the
reachable markings of $\FSI(N)$. Using this it is not hard to check
that implementations of contact-free nets are contact-free, and hence $\tau$-nets.

\begin{propositioni}{fsionesafe}{
  Let $N$ and $\FSI(N)$ be as before and $M \subseteq S \cup S^\tau$.
  }
    \item $M \in [M_0\rangle_{\FSI(N)}$ iff $\fsivalidmarking(M)$.
\label{fsi1}
    \item $\FSI(N)$ is contact-free.
\label{fsiimplstep2}
    \item $\FSI(N)$ is a $\tau$-net.
\label{fsi3}
\end{propositioni}
\begin{proof}
  \refitem{fsi1}:
  ``Only if'' follows from \reflem{fsiworking}\refitem{fsiimpllongstep},
   and ``if'' follows by Lemmas~\ref{lem-fsisim} and~\ref{lem-fsivalidmarking}.

  \refitem{fsiimplstep2}: Let $M\in[M_0\rangle_{\FSI(N)}$. Then
  $\fsivalidmarking(M)$, and hence $\tb(M) \in [M_0\rangle_N$.

 Consider any $t \in O$ with $\iprecond{t}\subseteq M$. Assume $(M
    \setminus \iprecond{t}) \cap \ipostcond{t} \ne \varnothing$. Since
    $\ipostcond{t} = \postcond{t} \subseteq S$ let $p \in S$ be such
    that $p \in M \cap \ipostcond{t}$ and $p \not\in \iprecond{t}$.
    As $N$ is contact-free we have $(\tb(M) \setminus \precond{t})
    \cap \postcond{t}=\emptyset$, so since $p \in \tb(M) \cap
    \postcond{t}$ it must be that $p \inp \precond{t}$.
    Hence $p_t \inp \iprecond{t} \subseteq M$ and we have $p \neq p_t$
    yet $\tb(p) \mathbin= p \mathbin= \tb(p_t)$, violating $\fsivalidmarking(M)$.

    Now consider any $t_p \in U'$ with $\iprecond{t_p}\subseteq M$.
    As $\iprecond{t_p} = \{p\}$ and $\ipostcond{t_p} = \{p_t\}$
    we have that $(M \setminus \iprecond{t_p}) \cap \ipostcond{t_p} \ne \varnothing$
    only if $p \in M \wedge p_t \in M$. However, $\tb(p) = p = \tb(p_t)$ which would
    violate $\fsivalidmarking(M)$.

  \refitem{fsi3}: By construction, $M_0$ is finite,
  $\iprecond{t}\neq\emptyset$, and $\iprecond{t}$ and $\ipostcond{t}$ are
  finite for all $t\inp O \cup U'$.
\end{proof}

By \reflem{fsiworking}\refitem{fsiimpltaustep} implementations are
always divergence free: 
\begin{proposition}{fsidivfree}{Let $N$ be a plain net.
  Then $\FSI(N)$ is divergence free.}
\vspace{-1.7em}\qed
\end{proposition}
\end{ICE2008-TR}

Whereas in a plain net $N$ for any sequence of observable transitions
$\sigma\in O^*$ there is at most one marking $M$ with $M_0
\Goesto[\sigma] M$, in its fully asynchronous
implementation $\FSI(N)$ there can be several such markings. These
markings $M'$ differ from $M$ in that some tokens may have wandered
off into the added
\begin{ICE2008}
invisible transitions on the incoming arcs of visible ones.
\end{ICE2008}
\begin{ICE2008-TR}
buffer places on the incoming arcs of visible transitions.
\end{ICE2008-TR}
As a consequence, a visible transition $t$ that is
enabled in $M$ need not be enabled in $M'$---we say that in
$\FSI(N)$ $t$ \defitem{can be refused after $\sigma$}.
This may occur for instance for the net $N$ of \reffig{fsi-fail},
namely with $\sigma=\varepsilon$ (the empty sequence), $M$ the initial
marking of $N$, $M'$ the marking of $\FSI(N)$ obtained by firing the
rightmost invisible transition, and $t=a$.

\begin{figure}
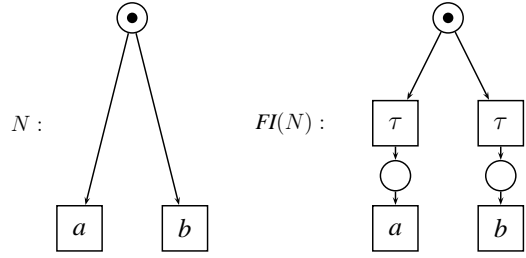

  \begin{center}
    \begin{petrinet}(11,5.5)
      \rput(1,3){\large $N:$}

      \P(3,5):np1;
      \t(2,1):nt1:a;
      \t(4,1):nt2:b;
      \a np1->nt1;
      \a np1->nt2;

      \rput(6,3){\large $\FSI(N):$}

      \P(9,5):fsip1;
      \t(8,3):fsit1p1:$\tau$;
      \p(8,2):fsip1t1;
      \t(8,1):fsit1:a;

      \t(10,3):fsit2p1:$\tau$;
      \p(10,2):fsip1t2;
      \t(10,1):fsit2:b;

      \a fsip1->fsit1p1; \a fsit1p1->fsip1t1; \a fsip1t1->fsit1;
      \a fsip1->fsit2p1; \a fsit2p1->fsip1t2; \a fsip1t2->fsit2;
    \end{petrinet}
  \end{center}
\vspace{-1em}
  \caption{A net which is not failures equivalent to its fully
  asynchronous implementation}
  \label{fig-fsi-fail}
\begin{ICE2008}
\vspace{1em}
\end{ICE2008}
\end{figure}

When this happens, we have $\fpair{\sigma,\{t\}} \in
\failureset(\FSI(N))\setminus\failureset(N)$, so the nets $N$ and $\FSI(N)$
are not failures equivalent.
\begin{ICE2008-TR}
The direction from implementation to
original is nicer however as every transition enabled in the implementation
must also have been enabled in the original net.
Hence the only difference in behaviour between original and implementation
can consist of additional failures in the implementation.

\begin{proposition}{fsiorigfailsub}{
  Let $N$ and $\FSI(N)$ be as before.
  Then $\failureset(N) \subseteq \failureset(\FSI(N))$.
  }
\vspace{-1.7em}
\end{proposition}
\begin{proof}
  Let $\fpair{\sigma,X} \inp \failureset(N)$. Applying
  \reflem{determ}, let $M_1 \subseteq S$ be the unique marking of $N$
  such that $M_0 \Production{\sigma}_N M_1$.
  By \reflem{fsisim} also $M_0 \Production{\sigma}_{\FSI(N)} M_1$.
  So $\fsivalidmarking(M_1)$.
  As $M_1\subseteq S$ we have $\tb(M_1)\!=\!M_1$.
  By \refpr{fsidivfree} there exists a marking $M_2$ with
  \plat{$M_1 \Production{\varepsilon}_{\FSI(N)} M_2 \wedge
  M_2 \arrownot\production{\tau}_{\FSI(N)}$}.
  \reflem{fsiworking}\refitem{fsiimpltaustep} yields
  $\tb(M_2) = \tb(M_1) \wedge \fsivalidmarking(M_2)$.

  Suppose $\fpair{\sigma,X} \not\in \failureset(\FSI(N))$.
  Then \plat{$M_2 \production{\{t\}}_{\FSI(N)} M_3$} for some $t\in X$ and marking
  $M_3$ of $\FSI(N)$.
  \reflem{fsiworking}\refitem{fsiimplstep} yields
  $M_1 = \tb(M_1) = \tb(M_2) \production{\{t\}} \tb(M_3)$, which is a contradiction.
\end{proof}
\end{ICE2008-TR}%
If\begin{ICE2008}, on the other hand,\end{ICE2008} the wandering off of tokens into
$\tau$-transitions never disables a transition that would be enabled
otherwise, then there is no essential behavioural difference between
$N$ and $\FSI(N)$, and they are equivalent in any reasonable
behavioural equivalence that abstracts from silent transition firings.
In that case, $N$ could be called \defitem{fully asynchronous}. 

\begin{definition}{fsa}{}
  The class of \defitem{fully asynchronous nets respecting
  branching time equivalence} is defined as\begin{ICE2008-TR}\\\end{ICE2008-TR}
  $\FSA(B) := \{N \mid \FSI(N) \approx_{\failureset} N\}$.
\end{definition}

As for any plain net $N$ we have $\failureset(N) \subseteq
\failureset(\FSI(N))$\begin{ICE2008} \cite{glabbeek08symmasymm}\end{ICE2008},
the class of nets $\FSA(B)$ can equivalently be defined as 
$\FSA(B) := \{N \mid \failureset(\FSI(N)) \subseteq \failureset(N)\}$.
\pagebreak[3]

It turns out that there exists a quite structural characterisation of those
nets which are failures equivalent to their fully asynchronous implementation.

\begin{ICE2008}
\begin{definition}{reachableconflict}{}
  A plain net $N = (S, O, \varnothing, F, M_0)$ \defitem{has a
  partially reachable conflict} iff $\exists t, u \inp O.\linebreak[1]
  t \ne u \wedge \precond{t} \cap \precond{u} \ne \varnothing$ and $\exists M
  \in [M_0\rangle. \precond{t} \subseteq M \vee \precond{u} \subseteq M$. 
\end{definition}
\end{ICE2008}
\begin{ICE2008-TR}
\begin{definition}{reachableconflict}{Let $N = (S, O, \varnothing, F, M_0)$ be a plain net}
  $N$ \defitem{has a
  partially reachable conflict} iff $\exists t, u \inp O.\linebreak[1]
  t \ne u \wedge \precond{t} \cap \precond{u} \ne \varnothing$ and $\exists M
  \in [M_0\rangle. \precond{t} \subseteq M$.
\end{definition}
\end{ICE2008-TR}

The nets $N$ of Figures~\ref{fig-fsi-fail} and~\ref{fig-si-deadlock}, for
instance, have a partially reachable conflict.

\begin{theorem}{rcfreeequalsfsa}
  {A plain net $N$ is in $\FSA(B)$ iff $N$ has no partially reachable conflict.}
\mbox{}\vspace{-1.4em}
\end{theorem}
\begin{ICE2008}
\begin{proof}
  See \cite{schicke08synchrony} or \cite{glabbeek08symmasymm}.
\end{proof}
\end{ICE2008}
\begin{ICE2008-TR}
\begin{proof}
  Let $N = (S, O, \varnothing, F, M_0)$ and $\FSI(N) = (S \cup S^\tau, O, U',
  F', M_0)$.

  ``$\Rightarrow$'':
  Assume $N$ has a partially reachable conflict. Then there exist $t, u\inp O$,
  $t \ne u$, $\sigma \inp O^*$ and $M_1 \subseteq S$ such that $M_0
  \Production{\sigma}_N M_1$, $\precond{t} \cap \precond{u} \ne \varnothing$ and
  $\precond{t} \subseteq M_1$.
By \reflem{determ} we know that $\fpair{\sigma,\{t\}} \not\in \failureset(N)$.

  On the other hand, $M_0 \Production{\sigma}_{\FSI(N)} M_1$ by \reflem{fsisim}. Let $p \in
  \precond{t} \cap \precond{u}$. Then, by construction of $\FSI(N)$, there
  exists an $M_2 \subseteq S \cup S^\tau$ with $M_1 [\{u_p\}\rangle M_2$, $p \mathbin{\notin}
  M_2$ and since $t \ne u$ also $p_t \mathbin{\notin} M_2$. Now let $M_3 \subseteq S \cup
  S^\tau$ such that \plat{$M_2 \production{\tau}^*_{\FSI(N)} M_3 \wedge M_3
  \arrownot\production{\tau}^*_{\FSI(N)}$} (which exists according to
  \refpr{fsidivfree}).\vspace{-3pt} Since $\forall v \in U'. p \notin \postcond{v} \wedge
  (p_t \in \postcond{v} \implies p \in \precond{v})$ we know that $p_t \notin M_3$.
  Thus $M_3 \arrownot\production{\{t\}}$ and there exists a failure pair
  $\fpair{\sigma,\{t\}} \in \failureset(\FSI(N))$. Hence $\failureset(\FSI(N))
  \ne \failureset(N)$, so $N \notin \FSA(B)$.

  ``$\Leftarrow$'':
  Assume $N \notin \FSA(B)$. Then $\failureset(\FSI(N)) \ne \failureset(N)$ and
  hence $\failureset(\FSI(N)) \setminus \failureset(N) \ne \varnothing$ by
  \refpr{fsiorigfailsub}.
  Let $\fpair{\sigma,X} \in \failureset(\FSI(N)) \setminus \failureset(N)$.
  Then there exists an $M_1 \subseteq S\cup S^\tau$ such that $M_0
  \Production{\sigma}_{\FSI(N)} M_1 \wedge M_1 \arrownot\production{\tau}
  \wedge \forall t \in X. M_1 \arrownot\production{\{t\}}$.
  By \reflem{fsiworking}\refitem{fsiimpllongstep} we have $M_0 \Production{\sigma}_N \tb(M_1)$.
\\
  Let $t \in X$ such that $\tb(M_1) \production{\{t\}}_N$ (which exists, otherwise
  $\fpair{\sigma, X} \in \failureset(N)$).\vspace{2pt} Let $p \in \precond{t}$ such that
  $p_t \notin M_1$ (such $p_t$ exists, otherwise \plat{$M_1 \production{\{t\}}_{\FSI(N)}$}).
  Since \plat{$\tb(M_1) \production{\{t\}}_N$} it follows that $p \in \tb(M_1)$. But $p
  \mathbin{\notin} M_1$, for otherwise $M_1 \production{\tau}_{\FSI(N)}$, which would be a
  contradiction. Hence there must exists some $u \in O$ with $p_u \in M_1$
  and $u \ne t$.  By construction of $\FSI(N)$ we have $p \in \precond{u}$.
  Thus $t, u \inp O \wedge t \ne u \wedge\linebreak[2]
  \precond{t} \cap \precond{u} \ne \varnothing
  \wedge \tb(M_1) \in [M_0\rangle_N \wedge \precond{t} \subseteq \tb(M_1)$ and
  $N$ has a partially reachable conflict.
\end{proof}
\end{ICE2008-TR}

\section{Symmetric Asynchrony}\label{sec-sa}

For investigating the next interaction pattern, we change our notion
of asynchronous implementation of a net. We only insert silent
transitions wherever a transition has multiple preplaces. These are
the situations where the synchronous removal of tokens is really essential.

\begin{definition}{si}{
  Let $N = (S, O, \varnothing, F, M_0)$ be a net.
  Let $O^b = \{t \mid t \in O, |\precond{t}| > 1\}$.
  }
  The \defitem{symmetrically asynchronous implementation} of $N$ is
  defined as the net\\
  $\SI(N) := (S \cup S^\tau, O, U', F', M_0)$ with\vspace{-1ex}
  \begin{align*}
    S^\tau :={\,} & \{\mathrlap{s_t}\hphantom{s_t} \mid t \in O^b, s \in \precond{t}\} \trail{,}\\
    U' :={\,} & \{\mathrlap{t_s}\hphantom{s_t} \mid t \in O^b, s \in \precond{t}\} \trail{and}\\
    F' :={\,} & F \cap \left((O \times S) \cup (S \times (O \setminus
      O^b))\right)\\
      &\cup \{(s, t_s), (t_s, s_t), (s_t, t) \mid t \in O^b, s \in \precond{t}\}
      \trail{.}
  \end{align*}
\end{definition}\vspace{-7pt}

An example is shown in \reffig{si-deadlock}.
\pagebreak[3]

\begin{figure}
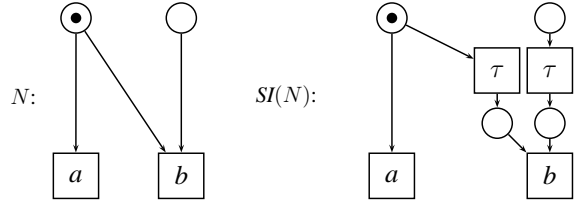

  \begin{center}
    \begin{petrinet}(13,4.5)
      \rput(1,2.5){\large $N$:}

      \P (2,4):p1;
      \p (4,4):p2;
      \t (2,1):t1:a;
      \t (4,1):t2:b;
      \a p1->t1;
      \a p1->t2;
      \a p2->t2;
      \a p2->n1;

      \rput(6,2.5){\large $\SI(N)$:}

      \P (8,4):p1b;
      \p (11,4):p2b;
      \t (11,3):p2btt:$\tau$;
      \p (11,2):p2btp;
      \t (10,3):p1bt2t:$\tau$;
      \p (10,2):p1bt2p;
      \t (8,1):t1b:a;
      \t (11,1):t2b:b;
      \a p1b->t1b;
      \a p2b->p2btt; \a p2btt->p2btp; \a p2btp->t2b;
      \a p1b->p1bt2t; \a p1bt2t->p1bt2p; \a p1bt2p->t2b;
      \a p2b->n1b;
    \end{petrinet}
  \end{center}
\vspace{-1em}
  \caption{The transition $a$ can be refused in $\SI(N)$ by firing the left $\tau$.}
  \label{fig-si-deadlock}
\end{figure}

\begin{ICE2008-TR}
Similar to Section \ref{sec-fsa}, we use $\iprecond{x}$ and
$\ipostcond{x}$ when describing the flow relation of the implementation.

As \refdf{si} is only a slight variation of \refdf{fsi}, the lemmas
and propositions about $\FSI$ in Section \ref{sec-fsa} apply to $\SI$
as well, with minimal changes in the proofs. We will again
begin with how the implementation can simulate the original net.

\begin{lemmai}{sisim}{
  Let $N = (S, O, \varnothing, F, M_0)$ be a plain net, $G\subseteq
  O$, $\sigma\in O^*$ and $M_1, M_2 \subseteq S$.
  }
  \item If $M_1 \production{G}_{N} M_2$ then
  $M_1 \production{\tau}^*_{\SI(N)}\production{G}_{\SI(N)} M_2$.
  \item If $M_1 \Production{\sigma}_{N} M_2$ then $M_1 \Production{\sigma}_{\SI(N)} M_2$.
\vspace{-1ex}
\end{lemmai}
\begin{proof}
  Let $O^b = \{t \mid t \in O, |\precond{t}| > 1\}$.
  Assume $M_1 \production{G}_N M_2$.  Then, by construction of $\SI(N)$,
  $$M_1 ~[\{t_s \mid t \inp G \cap O^b,~ s \inp \precond{t}\}\rangle_{\SI(N)}~
  [\{t \mid t \inp G\}\rangle_{\SI(N)}~ M_2.$$
  The rest of the proof is identical to the proof of \reflem{fsisim}.
\end{proof}

Also similar to the fully asynchronous case, we wish to
undo the effect of firing extraneous $\tau$-transitions. The function
doing so is the
same $\tb$ defined earlier. We also reuse the predicate $\fsivalidmarking$
and the distance function $\fsidistance$. However, $\fsidistance(M)$
is no longer a
\emph{strict} upper bound, or exact measure, on the number of silent
transitions that need to be fired from the marking $M$ before no
further silent transitions are possible. Optionally, strictness can be
ensured by replacing it by the function $\sidistance$, defined by
  $$\sidistance(M) := |M \cap \{s \mid s \inp S,~ \exists t \inp \postcond{s}.
  |\precond{t}| > 1\}|\trail{.}$$

Again $\fsivalidmarking(M)$ implies $\fsidistance(M) \in \IN$.

\begin{lemma}{sivalidmarking}{
  Let $N$ and $\SI(N)$ be as above and $M \subseteq S\cup S^\tau$,
  with $M$ finite.
  }
  Then $\forall p, q \inp M. \tb(p) = \tb(q) \implies p = q$ iff
  $\tb(M)\production{\tau}_{\SI(N)}^* M$.
\vspace{-1ex}
\end{lemma}
\begin{proof} 
  This is \reflem{fsivalidmarking} applied to $\SI(N)$ rather than $\FSI(N)$.
  The proof is identical.
\end{proof}

\begin{lemmai}{siworking}{
  Let $N$ and $\SI(N)$ be as above, $G\subseteq O$, $\sigma\in O^*$ and $M,M' \subseteq S \cup S^\tau$.
  }
    \item\label{siinvarstart} $\fsivalidmarking(M_0)$.
    \item\label{siimplstep} If $\fsivalidmarking(M) \wedge M
      \production{G}_{\SI(N)} M'$ then
      $\tb(M) \production{G}_N \tb(M') \wedge \fsivalidmarking(M')$.
    \item\label{siimpltaustep} If $\fsivalidmarking(M) \wedge M \production{\tau}_{\SI(N)} M'$
    then $\fsidistance(M) > \fsidistance(M') \wedge \tb(M) = \tb(M') \wedge \fsivalidmarking(M')$.
    \item\label{siimpllongstep} If $M_0 \Production{\sigma}_{\SI(N)}
      M'$ then $M_0 \Production{\sigma}_N \tb(M') \wedge \fsivalidmarking(M')$.
\vspace{-1em}
\end{lemmai}
\begin{proof}
  This is \reflem{fsiworking} applied to $\SI(N)$ rather than $\FSI(N)$;
  the proofs of \refitem{fsiinvarstart}, \refitem{siimpltaustep} and
  \refitem{siimpllongstep} are identical.

  \refitem{siimplstep}: Suppose $\fsivalidmarking(M)$ and $M
    \production{G}_{\SI(N)} M'$ with $G\subseteq O$.
    So $\tb(M)$ is a reachable marking of $N$.

    For any $t\in O$ and $s \in \precond{t}$ we set $\hat{s}_t := s_t$ if
    $|\precond{t}|>1$ and $\hat{s}_t := s$ otherwise.

    Let $t \in G$. Since $t$ is enabled in $M$, we have
    $\iprecond{t} \subseteq M$ and hence $\tb(\iprecond{t}) \subseteq
    \tb(M)$.  By construction, $\iprecond{t} = \{\hat{s}_t \mid s \in
    \precond{t}\}$ so $\tb(\iprecond{t}) = \precond{t}$.  Given that
    $N$ is contact-free, it follows that $t$ is enabled in $\tb(M)$.

    Now let $t,u \in G$ with $t\neq u$.  If $s \in \precond{t}
    \cap \precond{u}$ then $\hat{s}_t \in \iprecond{t}$ and $\hat{s}_u \in
    \iprecond{u}$, so $\hat{s}_t,\hat{s}_u \in M$. As $t$ and $u$ are
    independent in $\SI(N)$, we have $\hat{s}_t \ne \hat{s}_u$.
    However, $\tb(\hat{s}_t)=s=\tb(\hat{s}_u)$,
    contradicting $\fsivalidmarking(M)$. Hence $\precond{t} \cap
    \precond{u} = \emptyset$. Given that $\precond{t} \cup
    \precond{u} \subseteq \tb(M)$ and $N$ is contact-free, it
    follows that also $\postcond{t}\cap\postcond{u}=\emptyset$ and
    hence $t$ and $u$ are independent.

    We will now show that $\displaystyle
    \left(\tb(M)\setminus\bigcup_{t\in G} \precond{t}\right) \cup
        \bigcup_{t\in G}\postcond{t} = \tb(M')$.
    \begin{align*}
      M'
      ={}& (M \setminus \{s \mid s \in \iprecond{t},~ t \in G\}) \cup \{s \mid s \in
      \ipostcond{t},~ t \in G\}\\
      ={}& (M \setminus \{\hat{s}_t \mid s \in \precond{t},~ t \in G\})
      \cup \{ s \mid s \in \postcond{t},~ t \in G\}
      \trail{.}
    \end{align*}
    Therefore
       $\tb(M') = \tb(M \setminus \{\hat{s}_t \mid s \in \precond{t},~ t \in G\})
      \cup \tb(\{ s \mid s \in \postcond{t},~ t \in G\})$.

    Take any $t \inp G$ and any $s \inp \precond{t}$. Then $\hat{s}_t
    \inp M$, $\tb(\hat{s}_t)=s$ and $\fsivalidmarking(M)$ implies
    $\nexists p \inp M.\, p \ne \hat{s}_t \wedge \tb(p)=s$.
    Hence $\tb(M \setminus \{\hat{s}_t \mid s \in \precond{t},~ t \in G\}) =
    \tb(M) \setminus \{s \mid s \in \precond{t},~ t \in G\}$.
    Thus we find
    \begin{align*}
      \tb(M')
      ={}& \tb(M) \setminus \{s \mid s \in \precond{t},~ t \in G\}
      \cup \{ s \mid s \in \postcond{t},~ t \in G\}
    \end{align*}
    and conclude that $\tb(M) \production{G}_N \tb(M')$.

That $\fsivalidmarking(M')$ holds is established in exactly the same
way as in the proof of \reflem{fsiworking}\refitem{fsiimplstep},
noting that in deriving $p_t \in \iprecond{t} \subseteq M$ we use
$\hat{p}_t \in \iprecond{t} \subseteq M$ and $p\not\in M$.
\end{proof}

\begin{propositioni}{sionesafe}{
  Let $N$ and $\SI(N)$ be as before and $M \subseteq S \cup S^\tau$.
  }
    \item $M \in [M_0\rangle_{\SI(N)}$ iff $\fsivalidmarking(M)$.
\label{si1}
    \item $\SI(N)$ is contact-free.
\label{siimplstep2}
    \item $\SI(N)$ is a $\tau$-net.
\label{si3}
\vspace{-1ex}
\end{propositioni}
\begin{proof}
    Identical to the proof of \refpr{fsionesafe},
    using the lemmas of Section \ref{sec-sa}.
\end{proof}

\begin{proposition}{sidivfree}{Let $N$ be a plain net.
  Then $\SI(N)$ is divergence free.}
\vspace{-2em}
\end{proposition}

\begin{proof}
This follows immediately from \reflem{siworking}\refitem{fsiimpltaustep}.
\end{proof}
\end{ICE2008-TR}

\begin{ICE2008}
As for the fully asynchronous case,
an implementation of a plain net is always a divergence-free $\tau$-net.
\end{ICE2008}

\begin{ICE2008-TR}
\begin{proposition}{siorigfailsub}{
  Let $N$ and $\SI(N)$ be as before.
  Then $\failureset(N) \subseteq \failureset(\SI(N))$.
  }
\vspace{-2em}
\end{proposition}
\begin{proof}
  Identical to that of \refpr{fsiorigfailsub}, using the lemmas of
  Section \ref{sec-sa}.
\end{proof}
\end{ICE2008-TR}

Again, the only difference in behaviour between the original net and
its implementation is that observable transitions can potentially be
refused in the implementation, as in \reffig{si-deadlock}.
This yields a concept of a \defitem{symmetrically asynchronous} net.

\begin{definition}{sa}{}
    The class of \defitem{symmetrically asynchronous nets respecting
    branching time equivalence} is defined as
    $\SA(B) := \{N \mid \SI(N) \approx_{\failureset} N\}$.
\end{definition}

\begin{ICE2008}
Again we have $\failureset(N) \subseteq \failureset(\SI(N))$ for any
plain net $N$ \cite{glabbeek08symmasymm}.
\end{ICE2008}
We now show that plain nets can be implemented symmetrically asynchronously
with respect to failure equivalence exactly when they do not contain
reachable structures of the form shown in \reffig{si-deadlock}.

\begin{definition}{reachablen}{}
  A plain net $N = (S, O, \varnothing, F, M_0)$ \defitem{has a
  partially reachable \structuralN} iff $\exists t, u\inp O. t \ne u\linebreak[1]
  \wedge \precond{t} \cap \precond{u} \ne
  \varnothing\linebreak[1] \wedge |\precond{u}| > 1 \wedge
  \exists M \in [M_0\rangle_N. \precond{t} \subseteq M \vee \precond{u} \subseteq M$.
\end{definition}

\begin{theorem}{rnfreeequalssa}{A plain net $N$ is in $\SA(B)$ iff $N$ has no partially
  reachable \structuralN.}
\mbox{}\vspace{-1.65em}
\end{theorem}

\begin{ICE2008}
\begin{proof}
{
  See \cite{schicke08synchrony} or \cite{glabbeek08symmasymm}.
  \vspace{-1ex}
}
\end{proof}
\end{ICE2008}
\begin{ICE2008-TR}
\begin{proof}
  Let $N = (S, O, \varnothing, F, M_0)$ and $\SI(N) = (S \cup S^\tau, O, U',
  F', M_0)$.

  ``$\Rightarrow$'':
  Assume $N$ has a partially reachable \structuralN. Then there exist $t, u\inp O$,
  $t \ne u$, $\sigma \inp O^*$ and $M_1 \subseteq S$ such that $M_0
  \Production{\sigma}_N M_1$, $\precond{t} \cap \precond{u} \ne \varnothing$,
  $|\precond{u}| > 1$ and $\precond{t} \subseteq M_1 \vee \precond{u} \subseteq M_1$.
  We will show that $\SI(N) \not\approx_{\failureset} N$.

  There are two cases:

  Case 1, $\precond{t} \subseteq M$: We will show that $\fpair{\sigma,
  \{t\}} \inp \failureset(\SI(N))$ but $\fpair{\sigma, \{t\}}\mathbin{\not\in} \failureset(N)$.
  \vspace{-3pt}
  As $N$ has no silent transitions, by \reflem{determ} we have $M_0
  \Production{\sigma}_N M'$ only if $M' = M_1$. Since $M_1
  \production{\{t\}}_N$ it follows that $\fpair{\sigma, \{t\}} \notin \failureset(N)$.

  On the other hand, $M_0 \Production{\sigma}_{\SI(N)} M_1$ by \reflem{sisim}. Let $p \in
  \precond{t} \cap \precond{u}$. Then, by construction of $\SI(N)$, there
  exists an $M_2 \subseteq S \cup S^\tau$ with $M_1 [\{u_p\}\rangle M_2$, $p \mathbin{\notin}
  M_2$ and since $t \ne u$ also $\hat{p}_t \mathbin{\notin} M_2$. Now let $M_3 \subseteq S \cup
  S^\tau$ such that \plat{$M_2 \production{\tau}^*_{\SI(N)} M_3 \wedge M_3
  \arrownot\production{\tau}^*_{\SI(N)}$} (which exists according to
  \refpr{sidivfree}).\vspace{-3pt} Since $\forall v \inp U'. p \mathbin{\notin} \postcond{v} \wedge
  (p_t \inp \postcond{v} \implies p \inp \precond{v})$ we have $\hat{p}_t \mathbin{\notin} M_3$.
  Thus $M_3 \arrownot\production{\{t\}}$ and $\fpair{\sigma,\{t\}} \inp \failureset(\SI(N))$.

  Case 2, $\precond{t} \nsubseteq M$: Then $\precond{u} \subseteq M$. Thus
  $\exists q \in \precond{t} \setminus \precond{u}$, so $|\precond{t}| > 1$.
  This case proceeds as case 1 with the roles of $t$ and $u$ exchanged.

  ``$\Leftarrow$'':
  Assume $N \mathbin{\notin} \SA(B)$. Then $\failureset(\SI(N)) \ne \failureset(N)$ and
  hence $\failureset(\SI(N)) \setminus \failureset(N) \ne \varnothing$ by
  \refpr{siorigfailsub}.
  Let $\fpair{\sigma,X} \inp \failureset(\SI(N)) \setminus \failureset(N)$.
  \vspace{-3pt}
  Then there exists an $M_1 \subseteq S\cup S^\tau$ such that $M_0
  \Production{\sigma}_{\SI(N)} M_1$, $M_1 \arrownot\production{\tau}$
  and $\forall t \in X. M_1 \arrownot\production{\{t\}}$.
  By \reflem{siworking}\refitem{siimpllongstep} we have $M_0 \Production{\sigma}_N \tb(M_1)$.
\\
  Let $t \in X$ such that $\tb(M_1) \production{\{t\}}_N$ (which exists, otherwise
  $\fpair{\sigma, X} \in \failureset(N)$).\vspace{2pt} Let $p \in \precond{t}$ such that
  $\hat{p}_t \notin M_1$ (such a $p$ exists, otherwise \plat{$M_1 \production{\{t\}}_{\SI(N)}$}).
  Since \plat{$\tb(M_1) \production{\{t\}}_N$} it follows that $p \in \tb(M_1)$. But $p
  \mathbin{\notin} M_1$, for otherwise $p\ne\hat{p}_t$ and $M_1 \production{\tau}_{SI(N)}$,
  which would be a contradiction. Hence there must exists some $u \in O$ with $p_u \in M_1$
  and $u \ne t$.  By construction of $\SI(N)$ we have $p \in \precond{u}$ and $|\precond{u}| > 1$.
  Thus $t, u \inp O \wedge t \ne u \wedge\linebreak[2]
  \precond{t} \cap \precond{u} \ne \varnothing \wedge |\precond{u}| > 1
  \wedge \tb(M_1) \in [M_0\rangle_N \wedge \precond{t} \subseteq \tb(M_1)$, so
  $N$ has a partially reachable \structuralN.
\end{proof}
\end{ICE2008-TR}

The following proposition shows that the current class of nets strictly
extends the one from the previous section.

\begin{proposition}{fsabltsab}{$\FSA(B) \subsetneq \SA(B)$.}
\mbox{}\vspace{-1.65em}
\end{proposition}
\begin{proof}
  A net without partially reachable conflict surely has no partially
  reachable \structuralN.
  The inequality follows from the example in \reffig{fsi-fail}.
\vspace{-1ex}
\end{proof}

\begin{ICE2008}
\gdef\captiondistance{-1em}
\end{ICE2008}
\begin{ICE2008-TR}
\gdef\captiondistance{-0.3em}
\end{ICE2008-TR}

\def\figcounterexamplefclike{
\begin{figure}
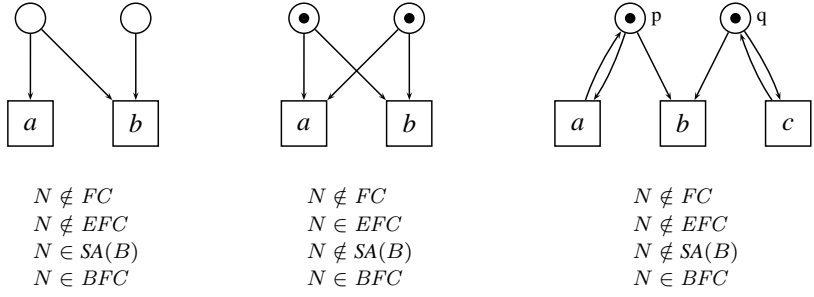

$$\begin{array}{c@{\qquad}c@{\qquad}c}
    \begin{petrinet}(4,3.4)
      \p (1,3):p1;
      \p (3,3):p2;
      \t (1,1):t1:a;
      \t (3,1):t2:b;

      \a p1->t1;
      \a p1->t2;
      \a p2->t2;
    \end{petrinet}
&
    \begin{petrinet}(4,3.4)
      \P (1,3):p1;
      \P (3,3):p2;
      \t (1,1):t1:a;
      \t (3,1):t2:b;

      \a p1->t1;
      \a p1->t2;
      \a p2->t1;
      \a p2->t2;
    \end{petrinet}
&
\begin{petrinet}(6,3.4)
  \Q (2,3):p1:p;
  \Q (4,3):p2:q;
  \t (1,1):t1:a;
  \t (3,1):t2:b;
  \t (5,1):t3:c;

  \a p1->t2;
  \a p2->t2;
  \A p1->t1;
  \A t1->p1;
  \A p2->t3;
  \A t3->p2;
\end{petrinet}\\[\captiondistance]
\alignedcaption[N \notin & \FC] &
\alignedcaption[N \notin & \FC] &
\alignedcaption[N \notin & \FC] \\[\captiondistance]
\alignedcaption[N \notin & \EFC] &
\alignedcaption[N \in & \EFC] &
\alignedcaption[N \notin & \EFC] \\[\captiondistance]
\alignedcaption[N \in & \SA(B)] &
\alignedcaption[N \notin & \SA(B)] &
\alignedcaption[N \notin & \SA(B)] \\[\captiondistance]
\alignedcaption[N \in & \BFC] &
\alignedcaption[N \in & \BFC] &
\alignedcaption[N \in & \BFC]
\end{array}$$
\caption{Differences between various classes of free-choice-like nets}
\label{fig-counterexamplesal_bfc}
\label{fig-counterexampleefc_sab}
\label{fig-counterexamplefcsal_efc}
\end{figure}
}

\begin{ICE2008-TR}
\figcounterexamplefclike
\end{ICE2008-TR}

It turns out that our class of nets $\SA(B)$ is strongly related to the
following established net classes \cite{bes87,best83freesimple}.

\begin{definition}{fcetc}{
  Let $N = (S, O, \varnothing, F, M_0)$ be a plain net.
  }
  \begin{enumerate}
    \item
      $N$ is \defitem{free choice}, $N \in \FC$, iff
      $\forall p, q \in S. p\neq q \wedge \postcond{p} \cap \postcond{q} \ne \varnothing
      \implies |\postcond{p}| = |\postcond{q}| = 1$.
    \item
      $N$ is \defitem{extended free choice}, $N \in \EFC$, iff
      $\forall p, q \in S. \postcond{p} \cap \postcond{q} \ne \varnothing
      \implies \postcond{p} = \postcond{q}$.
    \item
      $N$ is \defitem{behaviourally free choice}, $N \in \BFC$, iff
      $\forall u, v \in O. \precond{u} \cap \precond{v} \ne \varnothing
      \implies\\ \left(\forall M_1 \in [M_0\rangle.
      \precond{u} \subseteq M_1 \equivalent \precond{v} \subseteq M_1\right)$.
  \end{enumerate}
\end{definition}

The above definition of a free choice net is in terms of places, but
the notion can equivalently be defined in terms of transitions:\vspace{-1ex}
$$N \in \FC ~~\mbox{iff}~~
      \forall t, u \in T. t\neq u \wedge \precond{t} \cap \precond{u} \ne \varnothing
      \implies |\precond{t}| = |\precond{u}| = 1.\vspace{-1ex}$$
Both conditions are equivalent to the requirement that $N$ must be
$\structuralN$-free, where $\structuralN$ is defined as in \refdf{reachablen}
but without the reachability clause.
Also the notion of an extended free choice net can equivalently be
defined in terms of transitions:\vspace{-1ex}
$$N \in \EFC ~~\mbox{iff}~~
      \forall t, u \in T. \precond{t} \cap \precond{u} \ne \varnothing
      \implies \precond{t} = \precond{u}.\begin{ICE2008-TR}\vspace{-1ex}\end{ICE2008-TR}$$
This condition says that $N$ may not contain what we call a
\defitem{pure} $\structuralN$: places $p,q$ and transitions $t,u$ such that
$p\in\precond{t}\cap\precond{u}$, $q\in\precond{u}$ and $q\not\in\precond{t}$.

In \cite{best83freesimple} it has been established that
$\FC \subsetneq \EFC \subsetneq \BFC$. 
In fact, the inclusions follow directly from the definitions, and
\reffig{counterexamplefcsal_efc} displays counterexamples to strictness.

The class of free choice nets is strictly smaller than the class of symmetrically
asynchronous nets respecting branching time equivalence, which in turn
is strictly smaller than the class of behavioural free choice nets.
The class of extended free choice nets and the class of symmetrically
asynchronous nets respecting branching time equivalence are incomparable.
\begin{proposition}{fcltsab}{
  $\FC \subsetneq \SA(B) \subsetneq \BFC$,
  $\EFC \nsubseteq SA(B)$ and $SA(B) \nsubseteq \EFC$.}
\vspace{-2.15em}
\end{proposition}
\begin{proof}
The first inclusion follows because a partially reachable $\structuralN$ is
surely an $\structuralN$, and also the second inclusion follows
directly from the definitions.
  The four inequalities follow from the examples in
  \reffig{counterexampleefc_sab}. The first net is unmarked and thus trivially
  in $\SA(B)$. 
The second ones symmetrically asynchronous implementation has the additional
failure $\fpair{\varepsilon, \{a, b\}}$ and hence this net is not in $\SA(B)$.
\end{proof}
 
\begin{ICE2008}
\figcounterexamplefclike
\end{ICE2008}

In \reffig{symmetricasynorder} the relations between
our semantically defined net class $\SA(B)$, the structurally defined classes
$\FC$, $\EFC$, and the more behaviourally defined class $\BFC$ are summarised.
These relations may be interpreted as follows.

Starting at the top of the diagram, free choice nets are
characterised structurally, enforcing that for every place, a token therein can
choose freely (i.e.\ without inquiring about the existence of tokens in any
other places) which outgoing arc to take.
This property makes it possible to implement the system asynchronously. In
particular, the component which holds the information represented by a token
can choose arbitrarily when and into which of multiple asynchronous output
channels to forward said information, without further knowledge about the rest
of the system. As this decision is solely in the discretion of the sending
component and not based upon any knowledge of the rest of the system, no
synchronisation with other components is necessary.

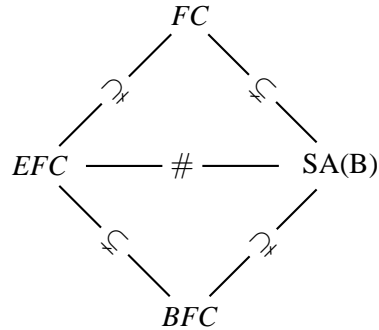
\begin{figure}[t]
  \begin{center}
    \begin{pspicture}(6,5.2)
      \rput(3,5){\ovalnode*{fc}{\FC}}
      \rput(1,3){\ovalnode*{efc}{\EFC}}
      \rput(5,3){\ovalnode*{sab}{SA(B)}}
      \rput(3,1){\ovalnode*{gfc}{\BFC}}

      \ncline{efc}{fc}\lput*{:U}{$\supsetneq$}
      \ncline{fc}{sab}\lput*{:U}{$\subsetneq$}
      \ncline{efc}{gfc}\lput*{:U}{$\subsetneq$}
      \ncline{gfc}{sab}\lput*{:U}{$\supsetneq$}
      \ncline{efc}{sab}\lput*{:U}{$\#$}
    \end{pspicture}
  \end{center}
\vspace{-2.5em}
  \caption{Overview of free-choice-like net classes}
  \label{fig-symmetricasynorder}
\vspace{.5em}
\end{figure}

The difference between $\SA(B)$ and $\FC$ is that in $\SA(B)$ the quantification
over the places is dropped, making the requirement more straightforward:
Every token can choose freely which outgoing arc to follow. Thus, $\SA(B)$
allows for non-free-choice structures as long as these never receive any
tokens.

This also explains why $\BFC$ includes $\SA(B)$. Since $\SA(B)$ guarantees that
all transitions of a problematic structure are never enabled, transitions in
such structures are never enabled while others are disabled.

The incomparability between the left and the right side of the diagram stems
from the conceptual allowance of slight transformations of the net before
evaluating whether it is free choice or not.
Extended free choice nets and behavioural free choice nets were
proposed as nets that are easily seen to be behaviourally equivalent to
free choice nets, and hence share some of their desirable properties:
in \cite{bes87,best83freesimple} constructions can be found to turn any
extended free choice net into an equivalent free choice net, and any
behavioural free choice net into an extended free choice net.%
{\footnotemark}
Applied on the last two nets in \reffig{counterexamplesal_bfc} these
constructions yield:
 
\begin{figure}[h]
\hfill
    \begin{petrinet}(4,3.4)
      \P (0,3):p1;
      \t (2,3):t3:$\tau$;
      \P (4,3):p2;
      \t (0,1):t1:a;
      \p (2,1):p3;
      \t (4,1):t2:b;

      \a p1->t3;
      \a p2->t3;
      \a t3->p3;
      \a p3->t1;
      \a p3->t2;
    \end{petrinet}
\hfill
\begin{petrinet}(6,4)
  \Q (2,3):p1:p;
  \Q (4,3):p2:q;
  \t (1,1):t1:a;
  \t (3,1):t2:b;
  \t (5,1):t3:c;

  \a p1->t2;
  \a p2->t2;
  \A p1->t1;
  \A t1->p1;
  \A p2->t1;
  \A t1->p2;
  \A p1->t3;
  \A t3->p1;
  \A p2->t3;
  \A t3->p2;
\end{petrinet}
\vspace{-1ex}
\hfill\mbox{}
\caption{Transformed nets from \reffig{counterexamplesal_bfc}}
\label{fig-transformed nets}
\vspace{-2pt}
\end{figure}

For the second net of \reffig{counterexamplesal_bfc}, a
$\tau$-transition is introduced, which collects both tokens and then
marks a single postplace from which the two original transitions are
enabled. Hence the choice between the two transitions is centralised
in the newly introduced place and thus free again. In the definition
of our symmetrically asynchronous implementation $\SI$, we do not
allow any insertion of such ``helping'' $\tau$-transitions, as it
seems unclear to us how much computing power should be allowed in
possibly larger networks of such transitions. This becomes especially
problematic if these networks somehow track part of the global status
of the net inside themselves and thus make quite informed decisions
about what outgoing transition to enable. 

\footnotetext{\label{BFC}
In \cite{bes87,best83freesimple} the nature of the equivalence between
the original and transformed net is not precisely specified. However, it
can be argued that whereas the transformation from \EFC-nets to
\FC-nets preserves branching time as well as causality, the
transformation from \BFC-nets to \EFC-nets preserves branching time
only: the third net of \reffig{counterexamplesal_bfc} is interleaving
bisimulation equivalent with its \EFC-counterpart in
\reffig{transformed nets}, but whereas the original net can perform
the transitions $a$ and $c$ concurrently (in one step), the transformed
net cannot.
}

\section{Asymmetric Asynchrony}\label{sec-aa}

As seen in the previous section, the class of symmetrically asynchronous nets
is quite small. It precludes the implementation of many real-world behaviours,
like waiting for one of multiple inputs to become readable, a Petri net
representation of which will always include non free-choice structures.

Therefore we propose a less strict definition of asynchrony such that
actions may depend synchronously on a single predetermined condition. In a
hardware implementation the places which earlier could always forward a
token into some silent transitions must now wait until they receive 
an explicit token removal signal from their posttransitions.

To this end we introduce a static priority over the preplaces of each transition.
Every transition first removes the token from the most prioritised preplace and
then continues along decreasing priority. To formalise this behaviour in a
Petri net we insert a silent transition for each incoming arc of every
transition. These silent transitions are forced to execute in sequence by
newly introduced buffer places between them. In the final position of this
chain, the original visible transition is executed.
An example of this transformation is given in \reffig{asym-transform}.

\begin{figure}
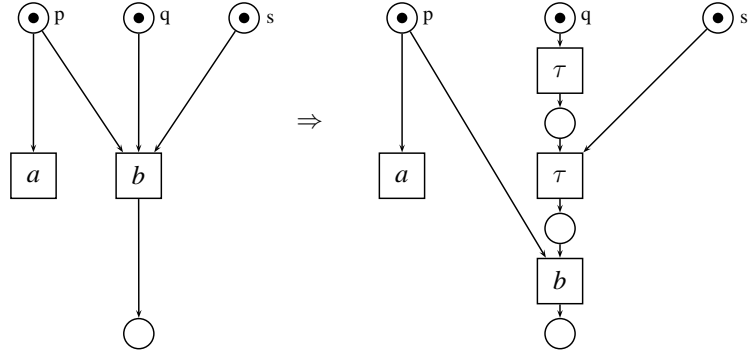

  \begin{center}
    \begin{petrinet}(14,7.5)
      \Q (1,7):p1:p;
      \Q (3,7):p2:q;
      \Q (5,7):p3:s;
      \p (3,1):p4;
      \t (1,4):t1:a;
      \t (3,4):t2:b;
      \a p1->t1;
      \a p1->t2;
      \a p2->t2;
      \a p3->t2;
      \a t2->p4;

      \rput(6.25,5){\Large $\Rightarrow$}

      \Q (8,7):p1c:p;
      \Q (11,7):p2c:q;
      \Q (14,7):p3c:s;
      \t (11,6):p2ctt:$\tau$;
      \p (11,5):p2ctp;
      \t (8,4):t1c:a;
      \t (11,4):t2c:$\tau$;
      \p (11,3):p3ctp;
      \t (11,2):t3ctt:b;
      \p (11,1):p4b;
      \a p1c->t1c;
      \a p2c->p2ctt; \a p2ctt->p2ctp; \a p2ctp->t2c;
      \a p1c->t3ctt;
      \a t2c->p3ctp;
      \a p3ctp->t3ctt;
      \a t3ctt->p4b;
      \a p3c->t2c;
    \end{petrinet}
  \end{center}
  \caption{Transformation to asymmetric asynchrony;
    $g$ such that $p <_g^b s <_g^b q$.
  }
  \label{fig-asym-transform}
\vspace{1em}
\end{figure}

\begin{ICE2008-TR}
\vspace{2em}
\pagebreak
\end{ICE2008-TR}

\begin{definition}{ai}{
  Let $N = (S, O, \varnothing, F, M_0)$ be a plain net.
  }
  Let $g \subseteq (S \times O) \times (S \times O)$ be a relation on $F \cap
  (S \times O)$ such that for each $t \in O$
\begin{ICE2008-TR}we have that\end{ICE2008-TR}
  $g \cap (\precond{t} \times \{t\})$ is a total order on
  $\precond{t} \times \{t\}$.
  Let $\leq_g^t$ be the total order on $\precond{t}$ given by
  $p \leq_g^t s$ iff $((p,t),(s,t)) \in g$.
  
  We write $\min_{g}^t$ for the $\leq_g^t$-minimal element of $\precond{t}$
  and $(s - 1)_g^t$ for the next place in $\precond{t}$ that is $\leq_g^t$-smaller than $s$.
  
  We define a set of silent transitions as $X := \{t_s \mid t \in O, s \in \precond{t}\}$.

  Let $h: X \rightarrow X \cup O$ be
  the function
  \begin{equation*}
    h(t_s) = \begin{cases}
      t & \text{ iff } s = \min_g^t\\
      t_s & \text{ otherwise }
    \end{cases}
\begin{ICE2008}
\vspace{2em}
\pagebreak
\end{ICE2008}%
  \end{equation*}
  The \defitem{asymmetrically asynchronous implementation with respect to
  $g$} of $N$ is defined as the net
  $\AI_{g} := (S \cup S^\tau, O, U', F', M_0)$ with\vspace{-1ex}
  \begin{align*}
    S^\tau :={\,} & \{s_t \mid t \in O,~ s \in \precond{t},~ s \ne \min_{g}^t\}
    \trail{,}\\
    U' :={\,} & h(X) \setminus O =
      \{t_s \mid t \in O,~ s \in \precond{t},~ s \ne \min_{g}^t\} \trail{and}\\
    F' :={\,} & F \cap (O \times S)\\
    & \cup \{\mathrlap{(s, h(t_s))}\hphantom{(p_t, h(t_s))} \mid
      t \in O,~ s \in \precond{t}\} \\
    & \cup \{\mathrlap{(t_s, s_t)}\hphantom{(p_t, h(t_s))} \mid
      t \in O,~ s \in \precond{t}, ~s \ne \min_{g}^t\} \\
    & \cup \{\mathrlap{(s_t, h(t_p))}\hphantom{(p_t, h(t_s))} \mid
      t \in O,~ s \in \precond{t}, ~s \ne \min_{g}^t, ~p = (s - 1)_g^t\}
    \trail{.}
  \end{align*}
\end{definition}\vspace{-1ex}

As before, we are interested in the relationship between nets
and their possible implementations. The definition of asymmetric
asynchrony however allows different implementations for the
same net.
\begin{ICE2008-TR}
We show that the lemmas and propositions from the previous sections
carry over for all possible implementations.
As in the earlier sections, we start
by showing how the implementation simulates the original.

\begin{lemmai}{aisim}{
  Let $N = (S, O, \varnothing, F, M_0)$ be a plain net, $G\subseteq
  O$, $\sigma\in O^*$, $M_1, M_2 \subseteq S$ and $g$ as above.
  }
  \item If $M_1 \production{G}_{N} M_2$ then
  $M_1 \production{\tau}^*_{\AI_{g}(N)}\production{G}_{\AI_{g}(N)} M_2$.
  \item If $M_1 \Production{\sigma}_{N} M_2$ then $M_1 \Production{\sigma}_{\AI_{g}(N)} M_2$.
\end{lemmai}
\begin{proof}
  Let $\AI_{g}(N) = (S \cup S^\tau, O, U', F', M_0)$.
  Assume $M_1 \production{G}_N M_2$. Then, due to the restrictions on $g$,
  there exists a sequence of pairwise disjoint nonempty sets $G_1, G_2, \ldots,
  G_n \subseteq U'$ such that
  $\forall t \inp G,\linebreak s \inp \precond{t},~ s \ne \min_g^t~
   \exists i,~ 1 \leq i \leq n.\, t_s \inp G_i$ and
  $\forall t_p \inp G_i,~ t_q \inp G_j,~ i < j.\, (t, p) \geq_g^t (t, q)$.
  By the construction of $\AI_{g}(N)$ then
  $$
  M_1 ~[G_1\rangle_{\AI_{g}(N)}
      ~[G_2\rangle_{\AI_{g}(N)}
      \cdots
      ~[G_n\rangle_{\AI_{g}(N)}
      ~[G\rangle_{\AI_{g}(N)}
      ~ M_2
  \trail{.}$$
  All non-final steps of that execution can be split into a sequence of
  singletons.
\\
  The second statement follows by a straightforward induction on the
  length of $\sigma$.
\end{proof}

As for the symmetrical case, we wish to push back all tokens on
$S^\tau$ in a marking to their roots in $S$. This time, however,
multiple silent transitions need to be undone.

\begin{definition}{taucompl}{
  Let $N = (S, O, \varnothing, F, M_0)$ be a plain net and
  let $\AI_{g}(N) = (S \cup S^\tau, O, U', F', M_0)$.
  }
  $\tbai: \powerset{S \cup S^\tau} \into \powerset{S}$ is the function
  defined by
  \begin{align*}
    \tbai(M) := (M \cap S)
    \cup \left\{s \mid \exists t \inp O.\, s \inp \precond{t} \wedge \exists
    p_t \inp M \cap S^\tau.\, p \leq_g^t s\right\}\trail{.}
  \end{align*}
\end{definition}

Given a reachable marking $M$ of the implementation, $\tbai$ will
produce a reachable marking of the original net, which by
\reflem{aisim} is also a reachable marking of the implementation, from
which $M$ could have arisen by firing some of the added unobservable
transitions.
Note that $\tbai(\iprecond{t}) = \precond{t}$\ for any $t \inp O$.
The application of $\tbai$ is only meaningful for markings where
no two elements of $S^\tau$ have originated from the same transition.
However, implementations of contact-free nets produce only reachable
markings which fulfil this condition, as we will show below.

We now give the invariant predicate $\aivalidmarking$ that characterises the
markings of an implementation that can be obtained from a reachable
marking of the original net by firing some unobservable transitions.

\begin{definition}{aivalidmarking}{
  Let $N = (S, O, \varnothing, F, M_0)$ be a plain net and $\AI_{g}(N)
  = (S \cup S^\tau, O, U', F', M_0)$.
  }
  The predicate $\aivalidmarking \subseteq \powerset{S \cup S^\tau}$ is given
  by
  $$\aivalidmarking(M) ~~:\equivalent~~
  \tbai(M) \in [M_0\rangle_N \wedge \forall p, q \in M.
  \tbai(\{p\}) \cap \tbai(\{q\}) \ne \varnothing \implies p = q
  \trail.$$
\end{definition}

Note that $\aivalidmarking(M)$ implies $\aidistance(M) \in \IN$.

\begin{lemma}{aivalidmarking}{
  Let $N$ and $\AI_{g}(N)$ be as above and $M \subseteq S\cup S^\tau$,
  with $M$ finite.
  }
  Then $\forall p, q \inp M.
  \tbai(\{p\}) \cap \tbai(\{q\}) \ne \varnothing \implies p = q$ iff
  $\tbai(M)\production{\tau}_{\AI_g}^* M$.
\end{lemma}
\begin{proof} 
  Given that $\tbai(M)\subseteq S$, ``if'' follows directly from the construction
  of $\AI_{g}(N)$.\\
  For ``only if'', assume $\forall p, q \inp M. \tbai(p) \cap \tbai(q) \ne \varnothing \implies p = q$.
  Let $G_1 := \{t_s \mid s_t \inp M\}$ and
  $G_{i+1} := \{t_p \mid \exists t_q \in G_i.\, q=(p-1)_g^t\}$ for $i>1$.
  The assumption guarantees that all $G_i$ are disjoint.
  Since $M$ is finite, and $\precond{t}$ is finite for all $t \inp O$,
  there must be an $n\geq 0$ such that $G_i=\emptyset$ iff $i>n$.
  Now $\tbai(M) ~[G_n\rangle_{\AI_{g}(N)}
                 ~[G_{n-1}\rangle_{\AI_{g}(N)}
	  \cdots ~[G_1\rangle_{\AI_{g}(N)}
	         ~M$.
\end{proof}

\begin{lemmai}{aiworking}{
  Let $N$ and $\AI_{g}(N)$ be as above, $G\subseteq O$, $\sigma\in
  O^*$ and $M,M' \subseteq S \cup S^\tau$.
  }
    \item\label{aiinvarstart} $\aivalidmarking(M_0)$.
    \item\label{aiimplstep} If $\aivalidmarking(M) \wedge M
      \production{G}_{\AI_g(N)} M'$ then
      $\tbai(M) \production{G}_N \tbai(M') \wedge \aivalidmarking(M')$.
    \item\label{aiimpltaustep} If $\aivalidmarking(M) \wedge M \production{\tau}_{\AI_g(N)} M'$
    then $\fsidistance(M) > \fsidistance(M') \wedge \tbai(M) = \tbai(M') \wedge \aivalidmarking(M')$.
    \item\label{aiimpllongstep} If $M_0 \Production{\sigma}_{\AI_g(N)}
      M'$ then $M_0 \Production{\sigma}_N \tbai(M') \wedge \aivalidmarking(M')$.
\end{lemmai}
\begin{proof}
  \refitem{aiinvarstart}: $M_0 \in [M_0\rangle_N$ and $\forall s \in M_0 \subseteq S. \tbai(\{s\}) = \{s\}$.

  \refitem{aiimplstep}:
    Suppose $\aivalidmarking(M)$ and $M \production{G}_{AI_{g}} M'$ with
    $G \subseteq O$. So $\tbai(M)$ is a reachable marking of $N$.

    Let $t \in G$. Since $t$ is enabled in $M$, we have $\iprecond{t} \subseteq
    M$ and hence $\tbai(\iprecond{t}) \subseteq \tbai(M)$. By construction,
    $\tbai(\iprecond{t}) = \precond{t}$. Given that $N$ is contact-free it
    follows that $t$ is enabled in $\tbai(M)$.

    Now let $t, u \in G$ with $t \ne u$. If $s \in \precond{t} \cap
    \precond{u}$ then $s \in \tbai(\iprecond{t}) \cap \tbai(\iprecond{u})$
    but since $t, u \in G$, $\iprecond{t} \cap \iprecond{u} =
    \varnothing$ and $\iprecond{t} \cup \iprecond{u} \subseteq M$,
    contradicting $\aivalidmarking(M)$. Hence $\precond{t} \cap \precond{u}
    = \varnothing$. Given that $\precond{t} \cup \precond{u} \subseteq
    \tbai(M)$ and $N$ is contact-free, it follows that also
    $\postcond{t} \cap \postcond{u} = \varnothing$ and hence $t$ and
    $u$ are independent.

    We will now show that $\displaystyle
    \left(\tbai(M)\setminus\bigcup_{t\in G} \precond{t}\right) \cup
        \bigcup_{t\in G}\postcond{t} = \tbai(M')$.

    By \refdf{taucompl} we have $\tbai(M) = \bigcup_{s\in M}
    \tbai(\{s\})$ for any $M \subseteq S\cup S^\tau$.
    Moreover, when $M$ satisfies $\aivalidmarking(M)$ this union is disjoint.
    In that case, for any set $Y \subseteq M$ we have
    $\tbai(M\setminus Y)=\tbai(M)\setminus\tbai(Y)$.
    \begin{align*}
      M'= (M \setminus \{s \mid s \in \iprecond{t},~ t \in G\}) \cup \{s \mid s \in
      \ipostcond{t},~ t \in G\} \trail.
    \end{align*}
    Therefore
    \begin{align*}
      \tbai(M')
      ={}& (\tbai(M) \setminus \tbai(\{s \mid s \in \iprecond{t},~ t \in G\}))
      \cup \tbai(\{s \mid s \in \ipostcond{t},~ t \in G\})\\
      ={}& (\tbai(M) \setminus
      \{s \mid s \inp \precond{t},~ t \inp G\})
      \cup \{s \mid s \inp \postcond{t},~ t \inp G\}\trail.
    \end{align*}

    Next we establish $\aivalidmarking(M')$. As in the proof
    of \reflem{fsiworking}\refitem{fsiimplstep} we may
    assume that $G$ is a singleton set $\{t\}$.
    Above we have shown that $\tbai(M')\in[M_0\rangle_N$.
    We still need to prove that $\tbai(\{p\}) \cap \tbai(\{q\}) \neq \varnothing
    \implies p = q$ for all $p, q \in M'$. Assume the contrary, i.e.\
    there are $p, q \inp M'$ with $\tbai(\{p\}) \cap \tbai(\{q\}) \ne \varnothing$
    but $p\ne q$.
    Since $\aivalidmarking(M)$, at least one of $p$ and $q$---say $p$---must
    not be present in $M$. Then $p \in \ipostcond{t} = \postcond{t} \subseteq
    S$. As $\tbai(\{q\}) \cap \tbai(\{p\}) \ne \varnothing$,
    $\tbai(\{p\}) = \{p\}$ and $q \ne p$ it must be that $q \inp S^\tau$. Hence
    $q \mathbin{\notin} \ipostcond{t}$, so $q \inp M$, and $p \in \tbai(\{q\})
    \subseteq \tbai(M)$. As shown above, $t$ is enabled in $\tbai(M)$. By the
    contact-freeness of $N$, $(\tbai(M) \setminus \precond{t}) \cap
    \postcond{t} = \varnothing$, so $p \inp \precond{t}$.
    Since $p\mathbin{\notin}M$, there exists a place $r_t \in
    (\iprecond{t} \cap S^\tau) \subseteq M$ with $p \in \tbai(\{r_t\})$.
    By construction, $\ipostcond{t} \cap S^\tau = \varnothing$, so we
    have $r_t \notin M'$, hence $q \ne r_t$.
    However, $p \in \tbai(\{q\}) \cap \tbai(\{r_t\})$,
    contradicting $\aivalidmarking(M)$.

  \refitem{aiimpltaustep}:
    Let $t_s \in U'$ such that $M [\{t_s\}\rangle_{\AI_{g}(N)} M'$.
    Then $\iprecond{t_s} \cap S = \{s\}$ and $\postcond{s}\ne\varnothing$.
    As $\ipostcond{t_s} \cap S = \varnothing$, no element of
    $\ipostcond{t_s}$ contributes to $\fsidistance(M')$ and
    hence $\fsidistance(M') = \fsidistance(M) - 1$.

    If $\iprecond{t_s} \subseteq S$ then
    $
      \tbai(M')
      = \tbai((M \setminus \iprecond{t_s}) \cup \ipostcond{t_s})
      = \tbai((M \setminus \{s\}) \cup \{s_t\})
      = \tbai(M)
    $.
    Otherwise let $p \in S$ such that $p_t \in \iprecond{t_s}$.
    Then
    $
      \tbai(M')
      = \tbai((M \setminus \iprecond{t_s}) \cup \ipostcond{t_s})
      = \tbai((M \setminus \{s, p_t\}) \cup \{s_t\})
      = \tbai(M)
    $.

    Moreover, $\aivalidmarking(M') \equivalent \aivalidmarking(M)$.

  \refitem{aiimpllongstep}: As in \reflem{fsiworking}.
\end{proof}

\begin{propositioni}{aionesafe}{
  Let $N$ and $\AI_{g}(N)$ be as before and $M \subseteq S \cup S^\tau$.
  }
  \item $M \in [M_0\rangle_{\AI_{g}(N)}$ iff $\aivalidmarking(M)$.
\label{ai1}
    \item $\AI_{g}(N)$ is contact-free.
\label{aiimplstep2}
    \item $\AI_{g}(N)$ is a $\tau$-net.
\label{ai3}
\end{propositioni}
\begin{proof}
  \refitem{ai1}: Identical to the proof of \refpr{fsionesafe}\refitem{fsi1},
    using the lemmas of Section \ref{sec-aa}.

  \refitem{aiimplstep2}: Let $M\in[M_0\rangle_{\AI_g(N)}$. Then
  $\aivalidmarking(M)$, and hence $\tbai(M) \in [M_0\rangle_N$.

 Consider any $t \in O$ with $\iprecond{t}\subseteq M$. Assume $(M
    \setminus \iprecond{t}) \cap \ipostcond{t} \ne \varnothing$. Since
    $\ipostcond{t} = \postcond{t} \subseteq S$ let $p \in S$ be such
    that $p \in M \cap \ipostcond{t}$ and $p \not\in \iprecond{t}$.
    As $N$ is contact-free we have $(\tbai(M) \setminus \precond{t})
    \cap \postcond{t}=\emptyset$, so since $p \in \tbai(M) \cap
    \postcond{t}$ it must be that $p \inp \precond{t}$.
    Hence, using that $p \not\in \iprecond{t}$ there must be an
    $s_t\inp\iprecond{t} \subseteq M$ with $s \leq_g^t p$,
    and hence $p\inp\tbai(\{s_t\})$. We have $p \neq s_t$
    yet $p\inp\tbai(\{p\})$, violating $\aivalidmarking(M)$.

    Now consider any $t_p \in U'$ with $\iprecond{t_p}\subseteq M$.
    As $p\inp \iprecond{t_p}$ and $\ipostcond{t_p} = \{p_t\}$
    we have that $(M \setminus \iprecond{t_p}) \cap \ipostcond{t_p} \ne \varnothing$
    only if $p \in M \wedge p_t \in M$.  However, $\tbai(\{p\}) = \{p\}
    \subseteq \tbai(\{p_t\})$ which would violate $\aivalidmarking(M)$.

    \refitem{ai3}: Identical to the proof of \refpr{fsionesafe}\refitem{fsi3}.
\end{proof}

\begin{proposition}{aidivfree}{Let $N$ be a plain net and $g$ as before.
   Then $\AI_{g}(N)$ is divergence free.}
\vspace{-2em}
\end{proposition}

\begin{proof}
This follows immediately from \reflem{aiworking}\refitem{aiimpltaustep}.
\pagebreak[2]
\end{proof}

\begin{proposition}{aiorigfailsub}{Let $N$ and $\AI_{g}(N)$ as before.
  Then $\failureset(N) \subseteq \failureset(\AI_{g}(N))$.}
\vspace{-2em}
\end{proposition}
\begin{proof}
  Identical to that of \refpr{fsiorigfailsub}, using the lemmas of
  Section \ref{sec-aa}.
\end{proof}

\end{ICE2008-TR}
We define a net to be \emph{asymmetrically asynchronous} if
any of the possible implementations simulates the net sufficiently.

\begin{definition}{aa}{}
  The class of \defitem{asymmetrically asynchronous nets respecting branching
  time equivalence} is
  defined as $\AA(B) := \{N \mid \exists g. \AI_{g}(N) \approx_{\failureset} N\}$.
\vspace{-3pt}
\end{definition}

As before,
\begin{ICE2008}
we have $\failureset(N) \subseteq \failureset(\AI_{g}(N))$ for any
plain net $N$ and any priority relation $g$ \cite{glabbeek08symmasymm}.
Additionally
\end{ICE2008}
we would like to obtain a semi-structural
characterisation of $\AA(B)$ in the spirit of
Theorems~\ref{thm-rcfreeequalsfsa} and~\ref{thm-rnfreeequalssa}.
Unfortunately we didn't succeed in this, but we obtained structural
upper and lower bounds for this net class.

\begin{definition}{reachablem}{}
  A net $N = (S, O, \varnothing, F, M_0)$ \defitem{has a left and right
  reachable \structuralM} iff
  $\exists t,u,v \inp O\;\exists p \inp \precond{t} \cap \precond{u}
  \linebreak[2]\; \exists q \inp
  \precond{u} \cap \precond{v}.\, t \ne u \wedge u \ne v \wedge p \ne q \wedge
  \exists M_1, M_2 \inp [M_0\rangle.
  \precond{t} \cup \precond{u} \subseteq M_1 \wedge
  \precond{v} \cup \precond{u} \subseteq M_2$.

  A net $N = (S, O, \varnothing, F, M_0)$ \defitem{has a left and right border
  reachable \structuralM} iff
  $\exists t,u,v \inp O\linebreak[2] \;\exists p \inp \precond{t} \cap \precond{u}
  \linebreak[2]\; \exists q \inp
  \precond{u} \cap \precond{v}.\, t \ne u \wedge u \ne v \wedge p \ne q \wedge
  \exists M_1, M_2 \inp [M_0\rangle.
  \precond{t} \subseteq M_1 \wedge
  \precond{v} \subseteq M_2$.
\end{definition}

\def\figureaahbnolrm{
\begin{figure}
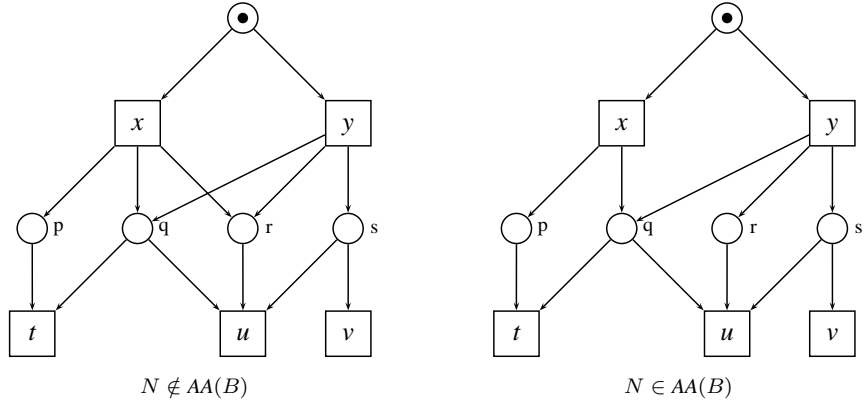

$$\begin{array}{c@{\qquad}c}
\begin{petrinet}(8,7.3)
  \q (1,3):p:p;
  \q (3,3):q:q;
  \q (5,3):r:r;
  \q (7,3):s:s;
  \P (5,7):i;

  \t (1,1):t:t;
  \t (5,1):u:u;
  \t (7,1):v:v;
  \t (3,5):x:x;
  \t (7,5):y:y;
  
  \a x->p; \a x->q; \a x->r;
  \a y->q; \a y->r; \a y->s;
  \a p->t; \a q->t;
  \a q->u; \a r->u; \a s->u;
  \a s->v;
  \a i->x; \a i->y;
\end{petrinet}
&
\begin{petrinet}(8,7.3)
  \q (1,3):p:p;
  \q (3,3):q:q;
  \q (5,3):r:r;
  \q (7,3):s:s;
  \P (5,7):i;

  \t (1,1):t:t;
  \t (5,1):u:u;
  \t (7,1):v:v;
  \t (3,5):x:x;
  \t (7,5):y:y;
  
  \a x->p; \a x->q;
  \a y->q; \a y->r; \a y->s;
  \a p->t; \a q->t;
  \a q->u; \a r->u; \a s->u;
  \a s->v;
  \a i->x; \a i->y;
\end{petrinet}\\[-1em]
\alignedcaption[N \notin & \AA(B)] &
\alignedcaption[N \in & \AA(B)]
\end{array}$$
\vspace{-1ex}
\caption{Nets which have a left and right border reachable \structuralM, but no left and right reachable \structuralM}
\label{fig-counterexampleaahb_nolrm}
\label{fig-counterexamplenolrm_aahb}
\end{figure}
}

\begin{ICE2008}
\figureaahbnolrm
\vspace{-0.6em}
\begin{theorem}{aahbimplieslrmfree}{}
  A plain net $N$ in $\AA(B)$ has no left and right reachable \structuralM.\\
  A plain net $N$ which has no left and right border
  reachable \structuralM{} is in $\AA(B)$.
\vspace{-1ex}
\end{theorem}
\begin{proof}
{
  See \cite{schicke08synchrony} or \cite{glabbeek08symmasymm}.
}
\vspace{-1ex}
\end{proof}
\end{ICE2008}

\begin{ICE2008-TR}
\begin{theorem}{aahbimplieslrmfree}{}
  A plain net $N$ in $\AA(B)$ has no left and right reachable \structuralM.
\end{theorem}
\begin{proof}
  Let $N = (S, O, \varnothing, F, M_0)$.
  Assume $N$ has a left and right reachable \structuralM.
  Then there exist $t, u, v \in O$ and $p, q \in S$ such that
  $p \in \precond{t} \cap \precond{u}$,
  $q \in \precond{u} \cap \precond{v}$, $t \ne u$, $u \ne v$, $p \ne q$
  and there are reachable markings $M_1, M_2 \inp [M_0\rangle$ such that
  $\precond{t} \cup \precond{u} \subseteq M_1$ and
  $\precond{v} \cup \precond{u} \subseteq M_2$.
  We will show that $\AI_g(N) \not\approx_{\failureset} N$, regardless
  of the choice of $g$.

  The problematic transition will be $u$. Either $p >_g^u q$ or
  $q >_g^u p$. Due to symmetry we can assume the former without
  loss of generality. So $p \ne \min_g^u$.
  We know that there is some $\sigma \inp O^*$ such that
  \plat{$M_0 \Production{\sigma}_N M_1 \wedge \precond{t} \subseteq M_1$}.
  By \reflem{determ} it follows that $\fpair{\sigma, \{t\}} \not\in
  \failureset(N)$.

  By \reflem{aisim} also $M_0 \Production{\sigma}_{\AI_{g}(N)} M_1$.
  Let $p_0, \ldots, p_n \in S$ such that
  $\{s \inp \precond{u} \mid p \leq_g^u s\} = \{p_0,\ldots,p_n\}$ with
  $p_0=p$ and $p_{i-1} = (p_i - 1)_g^u$ for $2 \leq i \leq n$.
  Since $\precond{u} \subseteq M_1$ there thus exists an $M_1'$ with
  $M_1 [\{u_{p_n}\}\rangle_{\AI_{g}(N)}[\{u_{p_{n-1}}\}\rangle_{\AI_{g}(N)}
  \cdots [\{u_{p_0}\}\rangle_{\AI_{g}(N)} M_1'$.
  Note that $p_u \in M_1'$.
  By Propositions \ref{pr-aidivfree} there exists an $M_1''$ with
  $M_1' \production{\tau}^*_{\AI_{g}(N)} M_1'' \wedge M_1''
  \arrownot\production{\tau}_{\AI_{g}(N)}$, and
  \refpr{aionesafe}\refitem{ai1} yields $\aivalidmarking(M_1'')$.

  From the construction of $\AI_{g}(N)$, using that $p_u \in M_1'$, it follows
  that $\exists s \in \precond{u}. s \leq_g^u p \wedge s_u \in M''_1$.
  Moreover, $p\in \tbai(\{s_u)\}$.
  We also have  $p \in \precond{t} = \tbai(\iprecond{t})$, so
  $\exists r \in \iprecond{t}.\, p \inp \tbai(\{r\})$.
  As $\aivalidmarking(M''_1)$, we have $r \not\in M''_1$, and thus
  $\iprecond{t} \nsubseteq M''_1$.
  Therefore $\fpair{\sigma, \{t\}} \in \failureset(\AI_{g}(N))$.
  Hence $N$ is not in $\AA(B)$.
\end{proof}

\begin{theorem}{weaklrmfreeimplesaahb}{}
  A plain net $N$ which has no left and right border
  reachable \structuralM{} is in $\AA(B)$.
\vspace{-1ex}
\end{theorem}
\begin{proof}
  Let $N = (S, O, \varnothing, F, M_0)$.
  Given a transition $u\inp O$, we say that a \emph{$u$-conflict
  occurs} in a preplace $p \inp \precond{u}$ when
  $\exists t \inp \postcond{p}. t \ne u \wedge
      (\exists M_1 \inp [M_0\rangle_N. \precond{t} \subseteq M_1))$.

  Assume $N$ has no left and right border reachable \structuralM.
  This means that every $u \in O$ has at most one preplace where an
  $u$-conflict occurs.
  Now choose $g \subseteq (S \times O) \times (S \times O)$ such that
  for all $u \inp O$, $\min_g^u$ is that single place, if it exists.
  Let $\AI_{g}(N) = (S \cup S^\tau, O, U', F', M_0)$.

  We prove that $\failureset(N) = \failureset(\AI_{g}(N))$.
  From \refpr{aiorigfailsub} we already have that
  $\failureset(N) \subseteq \failureset(\AI_{g}(N))$.
  Therefore consider a failure pair $\fpair{\sigma, X} \in \failureset(\AI_{g}(N))$.
  We need to show that $\fpair{\sigma, X} \in \failureset(N)$.

  There exists some $M_1 \subseteq S \cup S^\tau$ with \plat{$M_0
  \Production{\sigma}_{\AI_{g}(N)} M_1 \wedge M_1 \arrownot\production{\tau}
  \wedge \forall t \inp X. M_1 \arrownot\production{\{t\}}$}.\vspace{2pt}
  By \reflem{aiworking} $M_0 \Production{\sigma}_N \tbai(M_1)$.
  Now take any $t \in X$. Assume \plat{$\tbai(M_1) \production{\{t\}}_N$}.
  Then $\iprecond{t} \nsubseteq M_1$ but $\precond{t} \subseteq \tbai(M_1)$.
  By construction of $\tbai$ we have $\forall p \inp \precond{t}.\,
  p \inp M_1 \vee \exists u\inp O.\, p \inp \precond{u} \wedge \exists
  s_u \inp M_1 \cap S^\tau. s \leq_g^u p$.
  \\
  Now suppose we had $\forall p \inp \precond{t}.\,
  p \inp M_1 \vee \exists s_t \inp M_1 \cap S^\tau. s \leq_g^t p$.
  Then $\iprecond{t} \subseteq M_1$, using \plat{$M_1 \arrownot\production{\tau}$}.
  $\lightning$\linebreak
  Hence $\exists p \inp \precond{t}\,\exists u\inp O.\, u\ne t \wedge
  p \inp \precond{u} \wedge \exists s_u \inp M_1 \cap S^\tau. s \leq_g^u p$.

  But then $t \inp \postcond{p} \wedge t \ne u \wedge
  \tbai(M_1) \in [M_0\rangle_N \wedge \precond{t} \subseteq
  \tbai(M_1)$, so a $u$-conflict occurs in $p \inp \precond{u}$.
  Yet $\exists s \leq_g^u p.\, s_u\inp S^\tau$ implies that $s \ne
  \min_g^u$ and hence $p \ne \min_g^u$, by the construction of $\AI_{g}(N)$.
  This however contradicts our construction for $g$
  given above. Hence \plat{$\tbai(M_1) \arrownot\production{\{t\}}_N$}.
  Applying this argument for all $t \inp X$ yields
  $\fpair{\sigma, X} \inp \failureset(N)$ and thereby
  $\failureset(\AI_{g}(N)) \subseteq \failureset(N)$.
  Thus $N \in \AA(B)$.
\end{proof}
\figureaahbnolrm
\end{ICE2008-TR}

\reffig{counterexampleaahb_nolrm} shows two nets, each with
a left and right border reachable {\structuralM} but no left and right
reachable \structuralM, that thus fall in the grey area between our
structural upper and lower bounds for the class $\AA(B)$. In this case
the first net falls outside $\AA(B)$, whereas the second net falls
inside.  The crucial difference between these two examples is the
information available to $u$ about the execution of $y$.

There exists an implementation for the right net, namely by $u$ taking
the tokens from $r$, $q$ and $s$ in that order.  The first token (from
$r$) conveys the information that $y$ was executed, and thus $t$ is
not enabled.  Collecting the last token (from $s$) could fail, due to
$v$ removing it earlier.  Even so, removing the tokens from $r$ and
$q$ did not disable any transition that could fire in the original
net.  In the left net such an implementation will not work.

The following proposition says that our class of symmetrically
asynchronous nets strictly extends the corresponding class of
asymmetrically asynchronous nets.

\begin{proposition}{sabltaab}{
  $\SA(B) \subsetneq \AA(B)$.}
\vspace{-1.5em}
\end{proposition}
\begin{proof}
  A net which has no partially reachable \structuralN{} also has no left or right
  border reachable \structuralM.
  The inequality follows from the example in \reffig{si-deadlock}.
\vspace{-1ex}
\end{proof}

As before, our class $\AA(B)$ is related to some known net classes
\cite{best83freesimple}.

\begin{definition}{simpleetc}{
  Let $N = (S, O, \varnothing, F, M_0)$ be a plain net.
  }
  \begin{enumerate}
    \item
      $N$ is \defitem{simple}, $N \in \SPL$, iff $\forall p, q \inp S. p\neq q \wedge
      \postcond{p} \cap \postcond{q} \ne \varnothing \implies |\postcond{p}| =
      1 \vee |\postcond{q}| = 1$.
    \item
      $N$ is \defitem{extended simple}, $N \in \ESPL$, iff $\forall
      p, q\inp S. \postcond{p} \cap \postcond{q} \ne \varnothing \implies
      \postcond{p} \subseteq \postcond{q} \vee \postcond{q} \subseteq
      \postcond{p}$.
  \end{enumerate}
\end{definition}

Extended simple nets appear in \cite{bes87} under the name
\defitem{asymmetric choice systems}.
Note that simple is equivalent to $\structuralM$-free, where
$\structuralM$ is as in \refdf{reachablem} but without the reachability clauses.
Clearly, we have $\FC \subsetneq \SPL \subsetneq \ESPL$ and $\EFC
\subsetneq \ESPL$, whereas $\EFC \nsubseteq \SPL$ and $\SPL \nsubseteq \EFC$:
the inclusions follow immediately from the definitions, and
the first two nets of \reffig{counterexamplesal_bfc} provide
counterexamples to the inequalities.

The class of asymmetrically asynchronous nets respecting branching
time equivalence strictly extends the class of simple nets, whereas it
is incomparable with the class of extended simple nets.

\begin{ICE2008}
\begin{figure}[b]
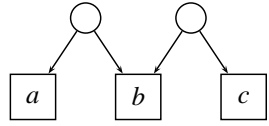

  \begin{center}
    \begin{petrinet}(6,3)
      \p (2,2.5):p1;
      \p (4,2.5):p2;
      \t (1,1):t1:a;
      \t (3,1):t2:b;
      \t (5,1):t3:c;
      \a p1->t1;
      \a p1->t2;
      \a p2->t2;
      \a p2->t3;
    \end{petrinet}
  \end{center}
\vspace{-1em}
  \caption{$N \in \AA(B)$, $N \notin \ESPL$}
  \label{fig-counterexampleespl_aahb}
\end{figure}
\end{ICE2008}
\begin{ICE2008-TR}
\begin{figure}[t]
  \begin{center}
    \begin{petrinet}(6,3)
      \p (2,2.5):p1;
      \p (4,2.5):p2;
      \t (1,1):t1:a;
      \t (3,1):t2:b;
      \t (5,1):t3:c;
      \a p1->t1;
      \a p1->t2;
      \a p2->t2;
      \a p2->t3;
    \end{petrinet}
  \end{center}
\vspace{-1.5em}
  \caption{$N \in \AA(B)$, $N \notin \ESPL$}
  \label{fig-counterexampleespl_aahb}
\vspace{-1.5em}
\end{figure}
\end{ICE2008-TR}

\begin{proposition}{aahbincompespl}{$\SPL \subsetneq \AA(B)$, $\AA(B)
    \nsubseteq \ESPL$ and $\ESPL \nsubseteq \AA(B)$.}
\vspace{-1.8em}
\end{proposition}
\begin{proof}
The inclusion is straightforward, and the inequalities
follow from the counterexamples in
  \reffig{counterexamplefcsal_efc} (the second one) and \reffig{counterexampleespl_aahb}.
  The missing tokens in the latter example are intended. As no action
  is possible there will not be any additional implementation failures.
\end{proof}

The relations between the classes $\SPL$, $\ESPL$ and $\AA(B)$ are
summarised in \reffig{asymmetricasynorder}.
Similarly to what we did in Section \ref{sec-sa}, we now try to translate
\reffig{asymmetricasynorder} into an intuitive description.

\begin{ICE2008}
\begin{figure}[t]
  \begin{center}
    \begin{pspicture}(5,3.2)
      \rput(2.5,3){\ovalnode*{spl}{\SPL}}
      \rput(1,1){\ovalnode*{espl}{\ESPL}}
      \rput(4,1){\ovalnode*{aahb}{AA(B)}}

      \ncline{espl}{spl}\lput*{:U}{$\supsetneq$}
      \ncline{spl}{aahb}\lput*{:U}{$\subsetneq$}
      \ncline{aahb}{espl}\lput*{:U}{$\#$}
    \end{pspicture}
  \end{center}
\vspace{-2em}
  \caption{Overview of asymmetric-choice-like net classes}
  \label{fig-asymmetricasynorder}
\end{figure}
\end{ICE2008}
\begin{ICE2008-TR}
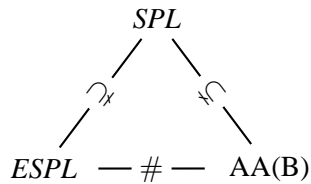
\begin{figure}[b]
  \begin{center}
    \begin{pspicture}(5,3.2)
      \rput(2.5,3){\ovalnode*{spl}{\SPL}}
      \rput(1,1){\ovalnode*{espl}{\ESPL}}
      \rput(4,1){\ovalnode*{aahb}{AA(B)}}

      \ncline{espl}{spl}\lput*{:U}{$\supsetneq$}
      \ncline{spl}{aahb}\lput*{:U}{$\subsetneq$}
      \ncline{aahb}{espl}\lput*{:U}{$\#$}
    \end{pspicture}
  \end{center}
\vspace{-2.5em}
  \caption{Overview of asymmetric-choice-like net classes}
  \label{fig-asymmetricasynorder}
\end{figure}
\end{ICE2008-TR}

The basic intuition behind $\SPL$ is that for every transition
there is only one preplace where conflict can possibly occur.
Whereas in $\SPL$ that possibility is determined by the static net structure,
in $\AA(B)$ reachability is also considered.

Similar to the difference between $FC$ and $EFC$ there exists a difference
between $\ESPL$ and $\SPL$ which originates from the fact that $\ESPL$
allows small transformations to a net before testing whether it lies in $\SPL$.
Again our class $\AA(B)$ does not allow such ``helping'' transformations.

\section{Conclusion and Related Work}\label{sec-conclusion}

We have investigated the effect of different types of asynchronous interaction,
using Petri nets as our system model. We propose three different interaction
patterns: fully asynchronous, symmetrically asynchronous and
asymmetrically asynchronous. An asynchronous implementation of a net is then
obtained by inserting silent (unobservable) transitions according to the
respective pattern. The pattern for asymmetric asynchrony is parametric in the
sense that the actual asynchronous implementation of a net depends on a chosen
priority function on the input places of a transition. For each of these cases,
we investigated for which types of nets the asynchronous implementation of a
net changes its behaviour with respect to failures equivalence (in the case of
asymmetric asynchrony, the `best' priority function may be used). It turns out
that we obtain a hierarchy of Petri net classes, where each class contains
those nets which do not change their behaviour when transformed into the
asynchronous version according to one of the interaction patterns.
This is not surprising because later constructions allow a more fine-grained
control over the interactions than earlier ones.

We did not consider connections from transitions to their postplaces
as relevant to determine asynchrony and distributability. This is
because we only discussed contact-free nets, where no synchronisation
by postplaces is necessary. In the spirit of \refdf{fsi} we could
insert $\tau$-transitions on any or all arcs from transitions to their
postplaces, and the resulting net would always be equivalent to the original.

Although we compare the behaviour of a net and its asynchronous
implementations in terms of failures equivalence, we believe that the
very same classes of nets are obtained when using any other reasonable
behavioural equivalence that respects branching time to some degree
and abstracts from silent transitions---no matter if this is an
interleaving equivalence, or one that respects causality.
We would get larger classes of nets, for example for the case of full
asynchrony including the net of \reffig{fsi-fail}, if we merely
required a net $N$ and its implementation to be equivalent under a
suitably chosen linear time equivalence.
This option is investigated in \cite{schicke08synchrony}.

The central results of the paper give semi-structural characterisations of our
semantically defined classes of nets. Moreover, we relate these classes to
well-known and well-understood structurally defined classes of nets, like free
choice nets, extended free choice nets and simple nets.   
  
To illustrate the potential interpretation of our results in other models of
distributed systems, we give an example.

\textit{Message sequence charts} (MSCs), also contained in UML 2.0 under the
name sequence diagrams, are a model for specifying
interactions between components (\textit{instances}) of a system. A
simple kind are \textit{basic message sequence charts} (BMSCs) as defined in
\cite{msc96}, where choices are not allowed. A Petri net semantics of BMCSs
with asynchronous communication and a unique sending and receiving event for each
message will yield Petri nets with unbranched places (see for instance \cite{graubmann93}).
Hence in this case the resulting Petri nets are conflict-free and therefore fully
asynchronously implementable according to \refthm{rcfreeequalsfsa}.

However in extended versions of MSCs, e.g.\ in UML 2.0 or in live sequence charts (LSCs, see \cite{harel03}), inline expressions allow to describe choices between
possible behaviours in MSCs. Consider for example the MSC given in \reffig{msc-example}
and a naive Petri net representation. The instances i1 and i2 can
either communicate or execute their local actions. Obviously, this requires
some mechanism in order to make sure that the choice is performed in a coherent
way (see e.g.\ \cite{gehrke99} for a discussion of this type of problem). In the
Petri net representation, we find a reachable \structuralN, hence with
\refthm{rnfreeequalssa} the net
does not belong to the class $\SA(B)$ of symmetrically asynchronously
implementable nets. However, the net is \structuralM-free, and thus does belong
to the class $\AA(B)$ of asymmetrically asynchronously implementable nets. By
giving priority to the collection of the message token (choosing the
appropriate function $g$ in our notion of implementation), it can be assured that
instance i2 does not make the wrong choice and gets stuck (however it is still
not clear whether the message will actually be consumed).

\begin{figure}
\hfill
    \begin{pspicture}(6,4.3)
      \rput(2,4){\rnode{obj1}{\psframebox{i1}}}
      \rput(5,4){\rnode{obj2}{\psframebox{i2}}}
      \psframe(1,0.8)(6,3.2)
      \psline[linestyle=dashed]{-}(1,2)(6,2)
      \psline{->}(2,2.6)(5,2.6)
      \rput[b](3.5,2.7){m}
      \rput(2,1.4){\rnode{actiona}{\psframebox{$a$}}}
      \rput(5,1.4){\rnode{actionb}{\psframebox{$b$}}}
      \pnode(2,0.4){end1}
      \pnode(5,0.4){end2}
      \ncline{-}{obj1}{actiona}
      \ncline{-}{obj2}{actionb}
      \ncline{-}{actiona}{end1}
      \ncline{-}{actionb}{end2}
      \psframe(1,2.7)(1.8,3.2)
      \rput(1.4,2.95){alt}
    \end{pspicture}
\hfill\hfill
    \begin{petrinet}(10,4)
      \P (3,5):p1;
      \P (7,5):p2;
      \p (5,3):p3;
      \p (3,1):p4;
      \p (7,1):p5;
      \t (1,3):a:a;
      \t (3,3):t1:\makebox[0pt]{\rm m!};
      \t (7,3):t2:\makebox[0pt]{\rm m?};
      \t (9,3):b:b;

      \a p1->a; \a p1->t1;
      \a p2->b; \a p2->t2;
      \a t1->p3; \a p3->t2;
      \a a->p4; \a t1->p4;
      \a b->p5; \a t2->p5;

      \rput(3,5.6){i1}
      \rput(7,5.6){i2}
    \end{petrinet}
\hfill\mbox{}
  \caption{An MSC and a potential implementation as Petri net, which has an \structuralN.}
  \label{fig-msc-example}
\end{figure}
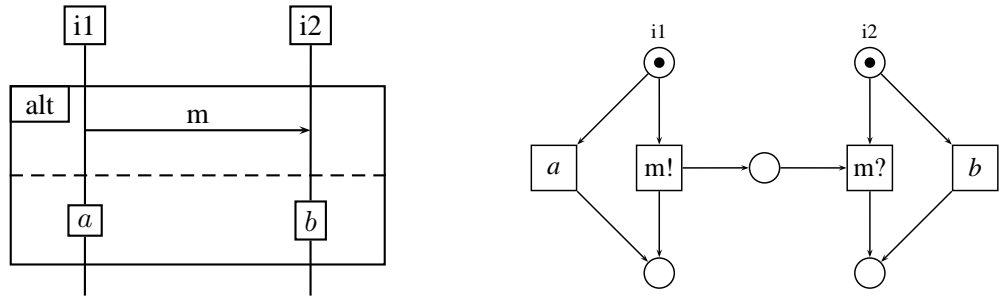

The obvious question is whether the naive Petri net interpretation we
have given is conform with the intended semantics of the
\textit{alt}-construct (according to the informal UML semantics the
alternatives always have to be executed completely; in LSCs it is
specified explicitly whether messages are assured to arrive). However,
on basis of a maybe more elaborate Petri nets semantics, it could be
discussed what types of MSCs can be used to describe physically
distributed systems, in particular which type of construct for choices
is reasonable in this case.

Another model of reactive systems where we can transfer our results to
are process algebras. When giving Petri net semantics to process
algebras, it is an interesting question to investigate which classes
of nets in our classification are obtained for certain types of
operators or restricted languages, and to compare the results with
results on language hierarchies (as summarised below).

We now give an overview on related work. A more extensive discussion
is contained in \cite{schicke08synchrony}.
We start by commenting on related work in Petri net theory. 

The structural net classes we compare our constructions to were all taken from
\cite{best83freesimple}, where Eike Best and Mike Shields introduce various
transformations between free choice nets, simple nets and extended variants
thereof. They use ``essential equivalence'' to compare the behaviour of
different nets, which they only give informally. This equivalence is
insensitive to divergence, which is also relied upon in their transformations.
As observed in Footnote~\ref{BFC}, it also does not preserve concurrency.
They continue to show conditions under which liveness can be
guaranteed for some of the classes.

In \cite{aalst98beyond}, Wil van der Aalst, Ekkart Kindler and J\"org Desel
introduce two extensions to extended simple nets, by allowing self-loops to
ignore the discipline imposed by the \ESPL-requirement.
This however assumes a kind of ``atomicity'' of
self-loops, which we did not allow in this paper. In particular we do not
implicitly assume that a transition will not change the state of a place it
is connected to by a self-loop, since in case of deadlock, the temporary
removal of a token from such a place might not be temporary indeed.

In \cite{reisig82buffersync} Wolfgang Reisig introduces a class of systems
which communicate using buffers and where the relative speeds of different
components are guaranteed to be irrelevant. The resulting nets are
simple nets. He then proceeds introducing a decision procedure for the problem
whether a marking exists which makes the complete system live.

The most similar work to our approach we have found is
\cite{hopkins91distnets}, where Richard Hopkins introduces the concept of
\emph{distributable} Petri Nets. These are defined in terms of
\emph{locality functions}, which assign to every transition $t$ a set of
possible machines or locations $L(t)$ on which $t$ may be executed,
subject to the restriction that a set of transitions with a common
preplace must share a common machine. A plain net $N$ is distributable
iff for every locality function $L$ that can be imposed on it, it has
a ``distributed implementation'', a $\tau$-net $N'$ with the same set
of visible transitions, in which each transition is assigned a
specific location, subject to three restrictions:\vspace{-6pt}
\begin{ICE2008-TR}\vspace{-2pt}\end{ICE2008-TR}\begin{itemize}
\item the location of a visible transition $t$ is chosen from $L(t)$,
\vspace{-3pt}\begin{ICE2008-TR}\vspace{-4pt}\end{ICE2008-TR}
\item transitions with a common preplace must have the same location
\vspace{-3pt}\begin{ICE2008-TR}\vspace{-4pt}\end{ICE2008-TR}
\item and there exists a weak bisimulation between $N$ and $N'$, such
that all $\tau$-transitions involved in simulating a transition $t$
from $N$ reside on one of the locations $L(t)$.
\vspace{-3pt}\begin{ICE2008-TR}\vspace{-4pt}\end{ICE2008-TR}
\end{itemize}
The last clause enforces both a behavioural correspondence between $N$
and $N'$ and a structural one (through the requirement on locations).
Thus, as in our work, the implementation is a $\tau$-net that is
required to be behaviourally equivalent to the original net.  However,
whereas we enforce particular implementations of an original net,
Hopkins allows implementations which are quite elaborate and make
informed decisions based upon global knowledge of the net.
Consequently, his class of distributable nets is larger than our
asynchronous net classes.  As Hopkins notes, due to his use of
interleaving semantics, his distributed implementations do not always
display the same concurrent behaviour as the original nets, namely
they add concurrency in some cases. This does not happen in our
asynchronous implementations.

Another branch of related work is in the context of distributed
algorithms. In \cite{bouge88symmetricleader} Luc Boug\'e considers the
problem of implementing symmetric leader election in the sublanguages
of CSP obtained by either allowing all guards, only input guards or no
communication guards at all in guarded choice. He finds that the
possibility of implementing it depends heavily on the structure of the
communication graphs, while truly symmetric schemes are only possible
in CSP with input and output guards.

Quite a number of papers consider the question of synchronous versus
asynchronous interaction in the realm of process algebras and the
$\pi$-calculus.
In \cite{boer91embedding} Frank de Boer and Catuscia Palamidessi consider
various dialects of CSP with differing degrees of asynchrony. In particular,
they consider CSP without output guards and CSP without any communication based
guards. They also consider explicitly asynchronous variants of CSP
where output actions cannot block, i.e.\ asynchronous sending is assumed.
Similar work is done for the $\pi$-calculus
in \cite{palamidessi97comparing} by Catuscia Palamidessi,
in \cite{nestmann00what} by Uwe Nestmann and
in \cite{G:FoSSaCS06} by Dianele Gorla.
A rich hierarchy of asynchronous $\pi$-calculi has been mapped
out in these papers. Again mixed-choice, i.e. the ability to combine
input and output guards in a single choice,
plays a central role in the
implementation of truly synchronous behaviour.
It would be interesting to explore the possible connections between
these languages and our net classes.

In \cite{selinger97firstorder}, Peter Selinger considers labelled
transition systems whose visible actions are partitioned into input and
output actions. He defines asynchronous implementations of such a
system by composing it with in- and output queues, and then
characterises the systems that are behaviourally equivalent to their
asynchronous implementations. The main difference with our approach is
that we focus on asynchrony within a system, whereas Selinger focusses
on the asynchronous nature of the communications of a system with the
outside world.

Finally, there are approaches on hardware design where asynchronous
interaction is an intriguing feature due to performance
issues. For this, see the papers \cite{lamport78ordering} and
\cite{lamport02arbitration} by Leslie Lamport. In
\cite{lamport02arbitration} he considers arbitration in
hardware and outlines various arbitration-free ``wait/signal'' registers.
He notes that nondeterminism is thought to require arbitration, but no proof
is known. He concludes that only marked graphs can be implemented using these
registers. Lamport then introduces ``Or-Waiting'', i.e.\ waiting for
any of two signals, but has no model available to characterise the resulting
processes. The used communication primitives bear a striking similarity to our
symmetrically asynchronous nets.

\end{document}